\documentclass[aps,prx,twocolumn,footinbib,superscriptaddress,floatfix]{revtex4-2}
\usepackage{graphicx}
\usepackage{indentfirst}
\usepackage{physics}
\usepackage{braket}
\usepackage{float}
\usepackage{amsmath}
\usepackage{amssymb}
\usepackage{mathtools}
\usepackage{epstopdf}
\usepackage{footnote}
\usepackage{CJK}
\usepackage{esint}
\usepackage{makecell}
\usepackage{comment}
\usepackage{color}
\usepackage[T1]{fontenc}
\usepackage{subfigure}
\usepackage{amsfonts}
\usepackage{footmisc}
\usepackage{scrextend}
\usepackage{diagbox}
\usepackage{multirow}
\usepackage{ulem}
\usepackage[hyperfootnotes=false]{hyperref}
\usepackage[acronym]{glossaries}
\usepackage[english]{babel}
\usepackage{url}
\usepackage{bm}
\usepackage{hyperref}
\usepackage{algpseudocode}
\usepackage{soul}
\usepackage{xcolor}
\definecolor{darkblue}{rgb}{0,0,0.5}
\hypersetup{
colorlinks=true,
linkcolor=black,
filecolor=blue,
citecolor=darkblue,  
urlcolor=black,
}

\usepackage{stmaryrd}
\usepackage{trimclip}

\makeatletter
\DeclareRobustCommand{\shortto}{%
  \mathrel{\mathpalette\short@to\relax}%
}

\newcommand{\short@to}[2]{%
  \mkern2mu
  \clipbox{{.5\width} 0 0 0}{$\m@th#1\vphantom{+}{\shortrightarrow}$}%
  }
\makeatother

\newtheorem{theorem}{Theorem}

\newtheorem{definition}[theorem]{Definition}

\newtheorem{lemma}[theorem]{Lemma}

\newtheorem{proposition}[theorem]{Proposition}

\newenvironment{proof}[1][Proof]{\noindent\textbf{#1.} }{\ \rule{0.5em}{0.5em}}

\newcommand{\1}{^{(1)}}

\def\d{{\rm d}}

\def\>{\rangle}
\def\<{\langle}

\newcommand{\map}[1]{\mathcal{#1}}

\newcommand{\QZ}[1]{{{\textcolor{black}{#1}}}}

\def\be{\begin{equation}}
\def\ee{\end{equation}}
\def\ba{\begin{eqnarray}}
\def\ea{\end{eqnarray}}

\begin{document}
\title{Theory Framework of Multiplexed Photon-Number-Resolving Detectors
}

\date{\today}

\author{Xiaobin Zhao}
\affiliation{Ming Hsieh Department of Electrical and Computer Engineering,
University of Southern California, Los Angeles, CA 90089, USA}

\author{Hezheng Qin}
\author{Hong X. Tang}
\affiliation{Department of Electrical Engineering, Yale University, New Haven, CT 06511, USA}

\author{Linran Fan}
\affiliation{Chandra Department of Electrical and Computer Engineering, The University of Texas at Austin, Austin, Texas 78758, USA}

\author{Quntao Zhuang}
\email{qzhuang@usc.edu}
\affiliation{Ming Hsieh Department of Electrical and Computer Engineering,
University of Southern California, Los Angeles, CA 90089, USA}
\affiliation{Department of Physics and Astronomy, University of Southern California, Los Angeles, CA 90089, USA}

\begin{abstract}

Photon counting is a fundamental component in quantum optics and quantum information. However, implementing ideal photon-number-resolving (PNR) detectors remains experimentally challenging. Multiplexed PNR detection offers a scalable and practical alternative by distributing photons across multiple modes and detecting their presence using simple ON-OFF detectors, thereby enabling approximate photon-number resolution. In this work, we establish a theoretical model for such detectors and prove that the estimation error in terms of photon number moments decreases inverse proportionally to the number of detectors.  
Thanks to the enhanced PNR capability, multiplexed PNR detector provides an advantage in cat-state breeding protocols. Assuming a two-photon subtraction case, $7$dB of squeezing, and an array of 20 detectors of efficiency $95\%$, our calculation predicts fidelity $\sim0.88$ with a success probability $\sim 3.8\%$, representing orders-of-magnitude improvement over previous works. Similar enhancement also extends to cat-state generation with the generalized photon number subtraction. With experimentally feasible parameters, our results suggest that megahertz-rate cat-state generation is achievable using an on-chip array of \emph{tens} of ON-OFF detectors.

\end{abstract}

\maketitle




\section{Introduction.}

Accurate photon counting lies at the heart of modern quantum-information science. Number-resolving detectors are indispensable for linear-optics quantum computing schemes~\cite{kok2007linear} and for demonstration of quantum advantages such as Boson-sampling~\cite{aaronson2011computational}. The same capability underpins long-distance quantum communication~\cite{gisin2007quantum} and secure key distribution protocols~\cite{gisin2002quantum}, as well as quantum-enhanced metrology, where it enables sub-shot-noise phase sensing~\cite{slussarenko2017unconditional} and optimal noise sensing~\cite{shi2023ultimate}. Beyond measurement, photon-number resolution makes it possible to  generate bosonic logical states \cite{vasconcelos2010all,weigand2018generating,konno2024logical,zeng2025neural} that serve as resources for fault-tolerant computation \cite{brady2024advances,noh2018quantum,gottesman2001encoding,michael2016new,fukui2023efficient} and capacity-achieving quantum communication \cite{harrington2001achievable,noh2018quantum}. 
Despite their pivotal role, accurately resolving the exact number of incident photons remains technically challenging. 
\QZ{
Previous superconducting photon-number-resolving detectors include calorimetric bolometers such as Microwave Kinetic Inductance Detectors (MKIDs) and transition-edge sensor (TESs) and superconducting nanowire single-photon detector (SNSPD)-based schemes \cite{day2003broadband,miller2003demonstration,cahall2017multi,reddy2020superconducting,kokkoniemi2020bolometer,zhu2020resolving}. TESs can resolve from a few up to several tens of photons per pulse, and recent designs reach of order one hundred photons, but they require sub-$100\,\mathrm{mK}$ operation and typically offer only kilohertz to sub-megahertz count rates \cite{eaton2023resolution,morais2024precisely,li2025boosting}. SNSPD-based photon-number-resolving schemes, on the other hand, can operate at higher temperatures and very high repetition rates, but single-nanowire or few-segment devices that exploit the nonlinear response of the nanowire have so far resolved only a few photons per pulse, around four to five \cite{cahall2017multi,zhu2020resolving}. These tradeoffs between dynamic range, cooling requirements and speed motivate alternative architectures based on spatio-temporal multiplexing.
}

Recently, spatio-temporally multiplexed photon-number–resolving (MPNR) detectors have emerged as a scalable route to approximate true number resolution: the optical field is split into many modes, each monitored by a binary ON-OFF sensor \cite{paul1996photon,divochiy2008superconducting,mattioli2016photon,schmidt2019characterization,wollman2019kilopixel,tiedau2019high,elshaari2020dispersion}. Building on this architecture, Cheng {\it et al.} have demonstrated an on-chip device that multiplexes superconducting nanowires along a single waveguide, achieving 100 pixels~\cite{cheng2023100}. Thanks to such advances, MPNR technology now combines (i) resolution of up to 100 pixels, (ii) operation at modest cryogenic temperatures ($\approx$ 2 K), (iii) gigahertz-level count rate, (iv) near-unity detection efficiency $\approx$ 95$\%$, (v)  sub-50 ns timing jitter, and (vi) dark-count rates of only a few hertz, features that make it a robust and versatile tool for quantum-optical experiments.

\begin{figure*}
\centering
\includegraphics[width=0.99\textwidth]{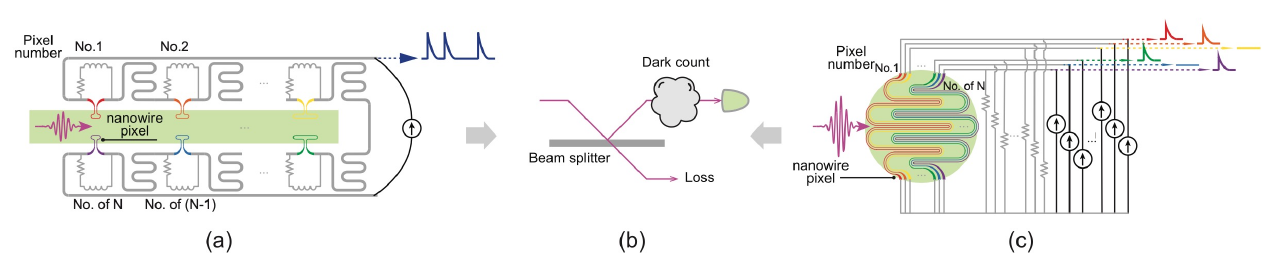} 
\caption{Main result of the paper. \QZ{ 
We develop a theoretical model that captures two realistic detection architectures, namely (a) a nanowire sequential detector that decodes pixel information from time-resolved electrical signals and (c) a parallel detector that decodes pixel information independently from separate electrical signals.}  As shown in (b), the theory model has an overall loss parameter and dark count parameter. \label{schematic_plot}}
\end{figure*}

In this work, we establish the theory framework for analyzing MPNR detectors. We evaluate the key performance metrics of MPNR detectors versus the number of multiplexed modes and demonstrate two applications: benchmarking photon-statistics reconstruction and generating non-Gaussian states. To begin with, we show the equivalence between two common architectures of MPNR detectors, the sequential and parallel schemes.
Further, we demonstrate that an MPNR detector with near-unit efficiency and binomially distributed dark counts can approximate a single PNR detector of comparable efficiency and dark counts. The approximation error is
$\mathcal O(1/n)$ for all the moments of the photon statistics, where $n$ is the number of ON-OFF detectors in multiplexing.  
As an example of photon statistics evaluation, we reconstruct the photon-number distribution of a squeezed vacuum from measurements taken with an MPNR detector. We introduce practical metrics to quantify the reconstruction error in this process. 

Moreover, we numerically analyze cat-state breeding with high-pixel MPNR detectors, charting the trade-off between fidelity and success probability as the detector size and efficiency vary. Assuming $7$dB of squeezing and an MPNR array of 20 detectors of efficiency $95\%$, we optimize cat state generation configuration, leading to an operation point with fidelity $\sim0.88$ and success probability $\sim 3.8\%$, representing orders-of-magnitude improvement over previous works~\cite{endo2023non}. With the experimental parameters in~\cite{cheng2023100}, our results suggest that megahertz rate cat-state generation is achievable using an on-chip MPNR array of tens of ON-OFF detectors. In addition, a comparable enhancement is observed in cat-state breeding when generalized photon-number subtraction is used.





\section{Model of multiplexed PNR detectors}


There are two main approaches to multiplexed PNR, as we depict in Fig.~\ref{schematic_plot} (a) and (c). In subplot (a), the incoming light goes through a sequence of nanowire detectors, where each detector splits a portion of the light to an ON-OFF detector. The photon number information is converted into a time-resolved electronic signal, as in Ref. \cite{cheng2023100}. In subplot (c), the input light is directly split into multiple portions, and each portion is fed into an ON-OFF detector, which registers a click or no click via an electronic signal, as in Refs.  \cite{liu2014photon,piacentini2015positive,grygar2022quantum}. \QZ{
From the viewpoint of our theory, these layouts are two realizations of the same multiplexed PNR model. In both cases, a single input mode is mapped by a linear optical network $U_{\rm BS}$ to $n$ ON–OFF detectors with splitting ratios $\{\eta_\ell\}$ and efficiencies $\{\kappa_\ell\}$. In the following discussion, Theorem~1 and all subsequent results are therefore formulated in terms of this unified beam-splitter model and apply equally to both architectures.
} 

To present the main theorem, we denote the portion of light split to each ON-OFF detector as $\eta_\ell$, the efficiency of each detector $\kappa_\ell$, and we assume uniform dark count $\epsilon\ll1$. In the 
configurations shown in Figs.~\ref{schematic_plot}(a) and (c), $\eta_\ell$ account for the cumulative effect of preceding splitting ratios. To make use of the large number of ON-OFF detectors, the design guarantees that $\eta_\ell\sim 1/n$, i.e., the incoming light is roughly evenly split to all ON-OFF detectors. \QZ{This near-uniform splitting is a standard design goal in time- and spatially-multiplexed PNR detectors and is well approximated in recent implementations~\cite{paul1996photon,divochiy2008superconducting,mattioli2016photon,schmidt2019characterization,wollman2019kilopixel,tiedau2019high,elshaari2020dispersion,cheng2023100}.} With these notations, we have our major results as follows. 
\begin{theorem}\label{theo1}
When detecting a quantum state with finite energy, a multiplexed PNR detector with ON-OFF detections can approximate an ideal PNR detector with efficiency $\overline{\kappa}=\sum_{j=1}^n \kappa_j\eta_j$ and binomial-distributed dark count with average $\overline{\epsilon}=n\epsilon$. The approximation error scales as $O(1/n)$ for all moments of the photon-number distribution.
\end{theorem}

By definition, the dark count of the detector is the photon counts on vacuum input. Therefore, the dark count can be obtained directly by summing up the uniformly distributed dark count of $n$ detectors, leading to a binomial-distributed dark count of average photon number $\overline{\epsilon}=n\epsilon$. \QZ{ For the case of non-uniform dark counts with $\{\epsilon_j \ll 1\}$, as long as $\sum_j \epsilon_j \ll 1$, our analysis still applies.
 } Now we can proceed to analyze the error from the ON-OFF detectors in terms of photon-number efficiency, where we can assume no dark count for simplicity. 

For the input mode $\hat a$ to the multiplexed PNR detector, the mode at the $j$-th ON-OFF detectors \QZ{is}
$
\hat{a}_j=\sqrt{\kappa_j \eta_j}\hat{a} +{\rm vac}
$,
where `vac' denotes vacuum modes. One can obtain this mode via first performing a single pure loss channel with transmissivity $\bar \kappa =\sum_{j=1}^n \kappa_j \eta_j$ and then a lossless beamsplitter array with the weights $\{\kappa_j \eta_j/\bar \kappa\}_{j=1}^n$. \QZ{ In general, we can combine the internal detection efficiency of the detector with the absorption efficiency into a single overall efficiency. Here, the 
efficiency is assumed to be unity for simplicity of analysis.
 }Therefore, the multiplexed PNR detector can be decomposed into a two-step process: In the first step, the input goes through a loss of $\bar \kappa$; In the second step, the mode goes through a lossless unit efficiency ON-OFF detector array, where a perfect beamsplitter array splits the mode into $n$ portions, one for each detector.  
Now we just need to show that a lossless network of unit efficiency ON-OFF detectors can approximate an ideal PNR detector, which is proven in the following lemma.
\begin{lemma}\label{theo:p_func_higher_moment} (lossless ON-OFF array approximating a PNR)
Given an arbitrary quantum state, applying a lossless $n$-plexed PNR detector without dark count to estimate the $h$-th moment of photon number ($h=\mathcal O(1)$ independent of $n$), the estimation error is:
\begin{align}
N^h_{\rm mpnr}-N^h&= \mathcal O \left(\frac {1} n\right), 
\end{align}
where $N^h=\<(\hat a^\dag \hat a)^h\>$ is the true value of the $h$-th order moment of photon number. 
\end{lemma}

\begin{figure}[t]
\centering\includegraphics[width=0.455\textwidth,angle=0,trim=0 0 0 0,clip]{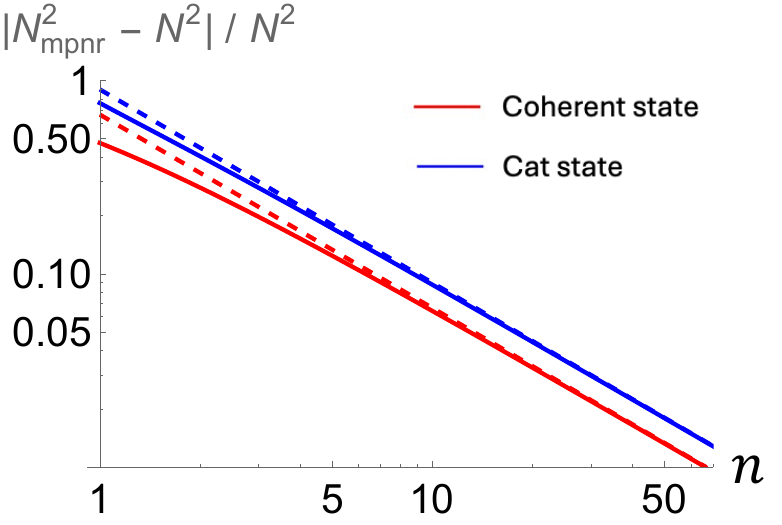}
\caption{Ratio of estimation error to the true value for higher moments of photon number distribution. The $x$-axis represents the number of ON-OFF detectors, and the $y$-axis shows the ratio of the estimation error, $\left|N^2_{\rm mpnr}-N_{\rm true}^2\right|$, to the true value $N^2$. Numerical results are presented for a coherent state $|\underline{\sqrt 2/2}\>$  (solid red line, with a scaling trend of $1/n$ shown as a red dashed line) and a cat state  $\propto |\underline{\sqrt 2/2}\>+|\underline{-\sqrt 2/2}\>$  (solid blue line, with a scaling trend of $1/n$ shown as a blue dashed line). Both axes are plotted on a logarithmic scale.  
}
\label{fig:cohernt_state}
\end{figure}

The concrete proof of Lemma  \ref{theo:p_func_higher_moment} can be found in Appendix \ref{app:proof of lemma 2}. 
In Fig. \ref{fig:cohernt_state}, we examine the error $\left|N^2_{\rm mpnr}-N^2\right|$ versus the true value of $N^2$ for two types of quantum states as an example.
First, we consider a coherent state, perhaps the most common quantum state in quantum optics. A coherent state with amplitude $\alpha$ is specified by the number basis wave function $|\underline{\alpha}\>:=\sum_{j=0}^\infty e^{-|\alpha|^2/2} \alpha^j/\sqrt{j!}|j\>$, where $\ket{j}$ denotes the Fock number state and we add underline in  coherent state to distinguish them from number states. To go beyond simple states, we also consider the Schrödinger cat states \cite{dakna1997generating,wenger2004pulsed,ourjoumtsev2006generating,lund2004conditional,takase2021generation,endo2023non,takase2023gottesman,takase2024generation}, $|\rm cat\>_{\alpha}^\pm \propto |\underline{\alpha}\>\pm|\underline{-\alpha}\>$, which are superposition of two coherent states. 
It can be seen that the scaling of the quantity $\left|N^2_{\rm mpnr}-N^2\right|/N^2$ converges to $\mathcal O(1 / n)$ as the number of detectors increases. 








\section{Measurement statistics in application examples}

For simplicity, let us  consider the MPNR detector operates through a balanced beamsplitter. 
Then, given an arbitrary pure state $|\psi\>$, the probability of obtaining a measurement outcome $k$ with an MPNR detector is given by $p_{k}= \<\psi'|P_{{\rm ON-OFF},k}|\psi'\>$ where $|\psi'\>=U_{\rm BS} |\psi\>|0\>^{\otimes n-1}$ is the state after applying a balanced beamsplitter. Here $P_{{\rm ON-OFF},k}=\bigoplus_{\sigma_g\in \map S_{n,k}}U_{\sigma_g} \left(\widetilde I^{\otimes k}\otimes |0\>\<0|^{\otimes n-k} \right) U_{\sigma_g}^\dag$ is a projector to the subspace associated with $k$ clicks, where $\map S_{n,k}$ represents the permutations corresponding to choosing $k$ out of $n$ objects, $U_{\sigma_g}$ is the unitary that implements the permutation between modes, \QZ{$\widetilde I=I-|0\>\<0|$ denotes the projector onto the subspace orthogonal to the vacuum state.} We will evaluate the performance of MPNR in applications.
To focus on the effects unique to MPNR, we will ignore dark count in the following analyses. A central quantum state examined in this section is the single-mode squeezed vacuum state, $\ket{r}_{\rm sq}=\frac{1}{\sqrt{\cosh r}} \sum_{n=0}^\infty (-\tanh r)^n \frac{\sqrt{(2n)!}}{2^n n!}|2n\>$, which only has occupation on even photon number states. Here $r$ is the squeezing amplitude, corresponding to squeezed quadrature variance $e^{-2r}$ below the vacuum noise level. Squeezed vacuums are versatile in quantum communication and quantum sensing, and also supply crucial resources in quantum computation.



\begin{figure}[t]
\centering
\includegraphics[width=0.48\textwidth,trim=2 12 2 2,clip,angle=0]{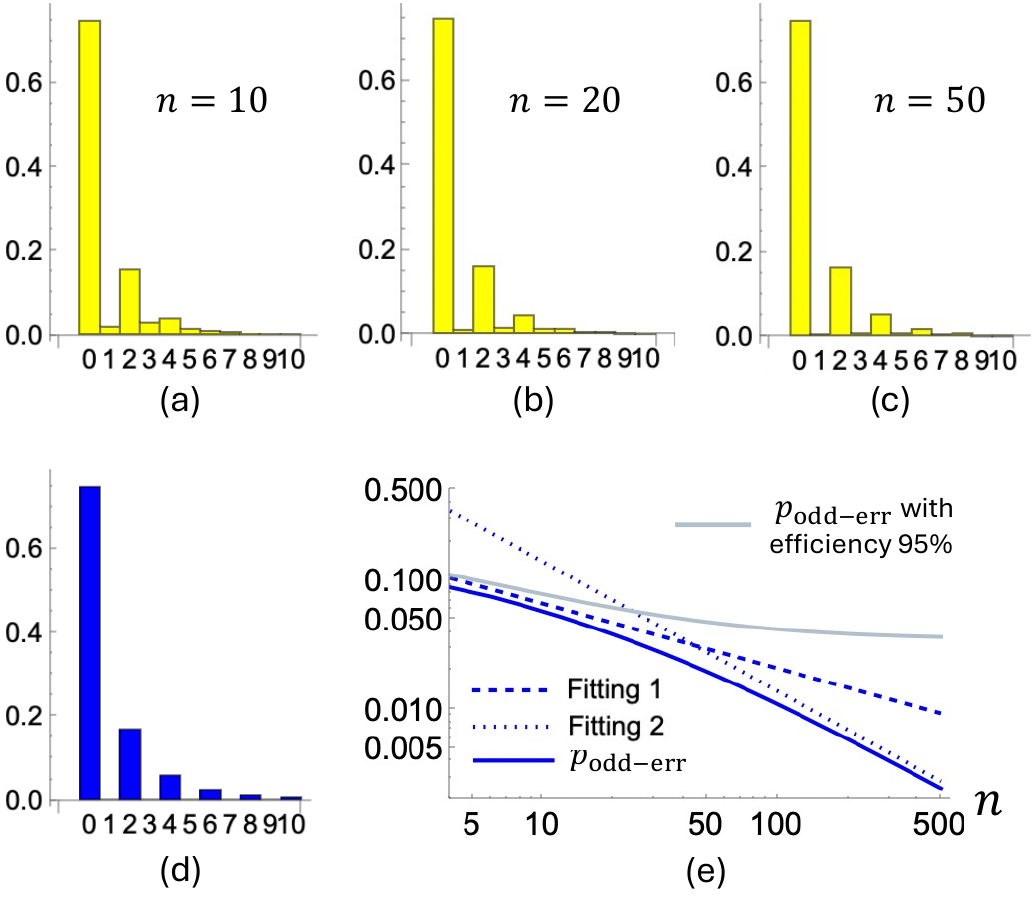}
\caption{Observed photon distributions for squeezed vacuum states. (a), (b), and (c) show the observed photon number distributions $\{p_k\}$ for ON-OFF detector numbers $n=10,20$, and $50$, respectively, for a single-mode squeezed vacuum state with a squeezing level of $7$dB. (d) Displays the true photon number distribution. (e) Depicts the odd-photon error probability  $p_{\rm odd-err}$ where the solid blue line represents calculated values, the solid light blue line denotes the error probability $p_{\rm odd-err}$ with detector efficiency $95\%$, the  dashed blue line indicates a fitting function $\sim\mathcal O(n^{-0.5})$,  the dotted blue line corresponds to $\sim\mathcal O(n^{-1})$.    
}
\label{fig:calibration_squeezing}
\end{figure}

\subsection{Detecting squeezed vacuum photon distribution}

To begin with, we consider the photon statistics of the single-mode squeezed vacuum on an MPNR detector. Due to the special photon number statistics of squeezed vacuum, it serves as a potential approach of benchmarking MPNR detectors. To evaluate the performance of photon-statistics reconstruction, let us first consider unit detection efficiency and neglect dark counts.  Explicitly, the following proposition is useful:

\begin{proposition}[Measuring squeezed states]\label{pro:pnr_cali}
With a noiseless balanced MPNR detector (without loss and dark count), the observed photon number distribution of a squeezed vacuum state $|r\>_{\rm sq}$ with squeezing amplitude $r$ is: 
\begin{align}
p_{{\rm sq},k}
= & \frac{n!}{k!(n-k)!} \sum_{\ell=0}^\infty \frac{1}{\cosh r} \frac{(\tanh r)^{2\ell}  }{2^{2\ell}(\ell!)^2} \frac{[(2\ell)!]^2}{n^{2\ell}}\nonumber \\ 
&\times \sum_{j_{1}+\cdots+j_k=2\ell\atop j_{1},\cdots,j_k=1,\cdots,2\ell}  \frac{1}{j_{1}!\cdots j_k!} .
\end{align}
\end{proposition}
The detailed proof of Proposition \ref{pro:pnr_cali} is shown in Appendix \ref{app:calibration_squeezied state}.
In Figs.~\ref{fig:calibration_squeezing} (a)-(d), we present a numerical calculation comparing the true photon distribution (subplot d) to the observed photon distribution (subplots a,b,c, \QZ{ precise up to the truncation error of the Hilbert space}). As shown, the probability of observing odd-photon numbers decreases as the number of ON-OFF detectors increases. 

Since the ideal squeezed vacuum consists solely of even-photon-number components, the detection of odd photon number can capture the non-ideality of the detector. In this regard, a simple way to quantify the performance is to evaluate the probability \QZ{$p_{{\rm odd-err}}:=\sum_{j=0}^\infty p_{{\rm sq},2j+1}$} of obtaining odd photon number.
A numerical calculation of the probability $p_{\rm odd-err}$ is presented in Fig. \ref{fig:calibration_squeezing}(e). The results demonstrate that the odd-photon error probability for squeezed vacuum states decreases at a rate faster than $n^{-0.5}$ but slower than $n^{-1}$ with $n$ being the number of ON-OFF detectors in the MPNR system. \QZ{ In addition to the lossless case, we also simulate a realistic scenario
where each ON--OFF detector has an efficiency of \(95\%\). As expected,
the resulting error probability \(p_{\rm odd\text{-}err}\) converges to a constant as the number of ON--OFF detectors increases. } Note that this error probability is applicable to various foundational quantum states, including even Schr\"{o}dinger cat states and squeezed even Fock states. 
Furthermore, for states with support only on odd photon numbers, an analogous “even-photon error” can be similarly defined.

\QZ{ Note that the odd-number probability $p_{\rm odd-err}$ converges to $\mathcal O(1/n)$, i.e., 
\begin{align}
\left. p_{\rm odd-err}\right|_{n\gg 1}=\mathcal O\left(\frac 1 n \right)
\end{align}
due to the relations $2\ell \ge k+1$ when $k$ is odd, $n!/(n-k)!=\mathcal O(n^k)$, and the fact that the summand 
\begin{align}
\sum_{k\ {\rm is\ odd} }\sum_{\ell=0}^\infty \frac{(\tanh r)^{2\ell}  }{\cosh r 2^{2\ell}(\ell!)^2[(2\ell)!]^2} \sum_{j_{1}+\cdots+j_k=2\ell\atop j_{1},\cdots,j_k=1,\cdots,2\ell}  \frac{1}{j_{1}!\cdots j_k!}
\end{align}
is not a function of $n$. 
}

\begin{figure}
\centering\includegraphics[width=0.47\textwidth,trim=5 5 5 5,clip,angle=0]{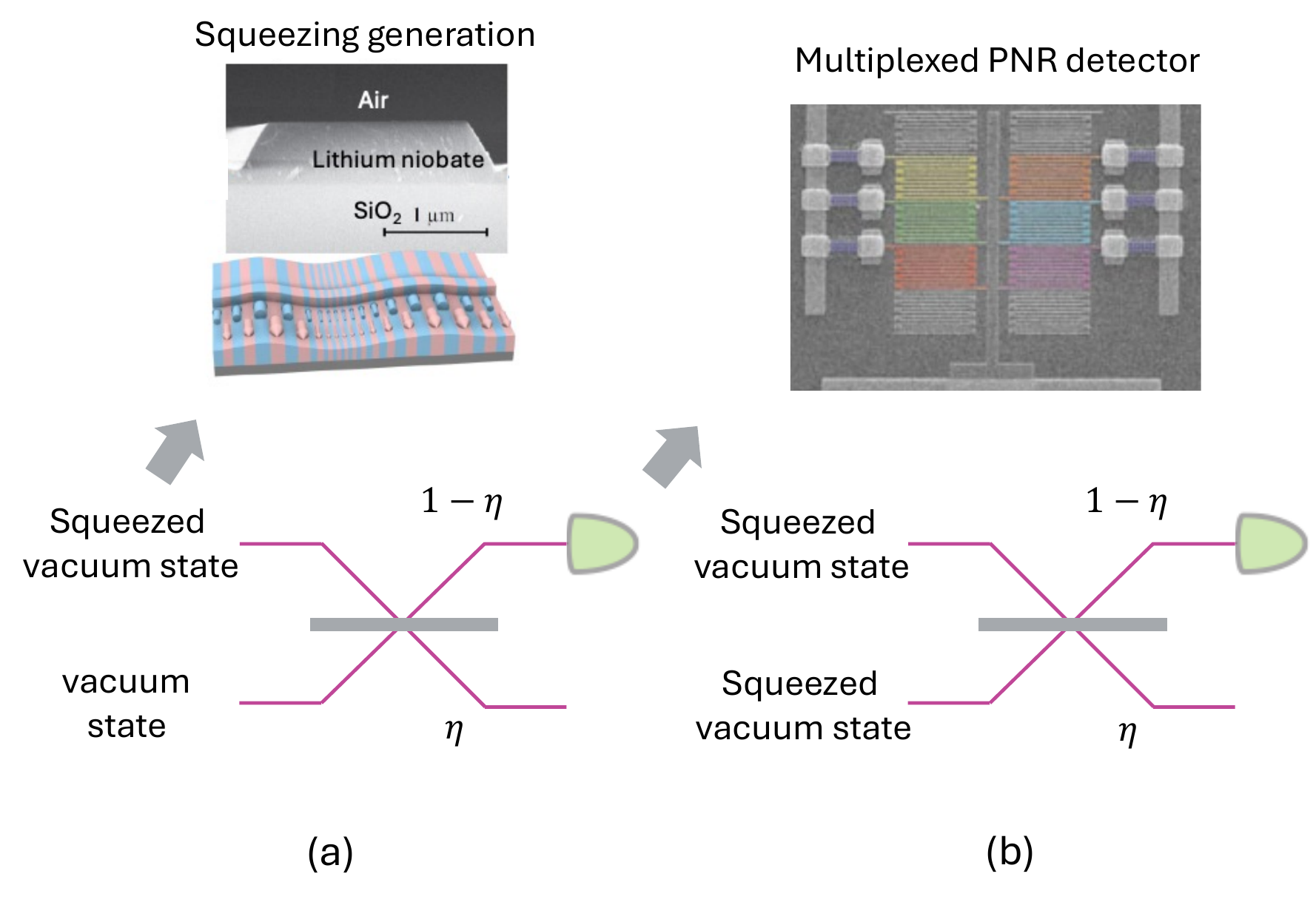}
\caption{Schematic of the cat state generation protocol with MPNR detection. (a) Photon subtraction (b) Generalized photon subtraction. The beamsplitter splits $\eta$ portion of light towards the cat state output. The MPNR detector figure is adopted from Ref.~\cite{divochiy2008superconducting}. }
\label{fig:scheme_cat}
\end{figure}

\begin{figure*}
\centering\includegraphics[width=0.73\textwidth,trim=2 0 0 2,clip,angle=0]{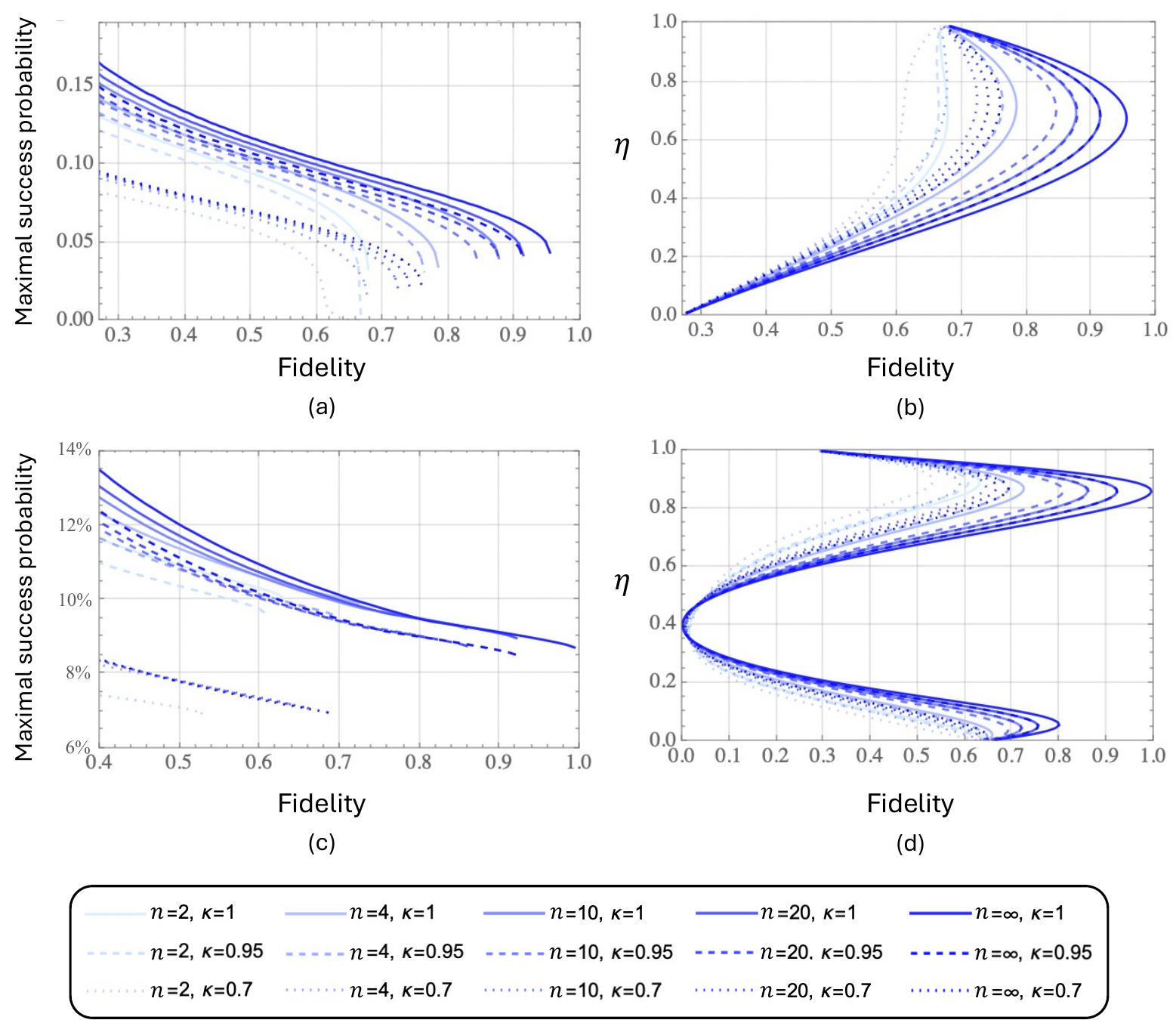}
\caption{Impact of detector number $n$ and efficiency $\kappa$ on cat-state breeding. 
\textbf{(a,c)} Maximum success probability versus fidelity, obtained by tuning the transmissivity $\eta$ (in the range $0 \le \eta \le 1$ for (a) and $0.4 \le \eta \le 1$ for (c)). \textbf{(b,d)} transmissivity values $\eta$ corresponding to a given  fidelity. Panels (a) and (b) present results for photon subtraction, whereas panels (c) and (d) illustrate generalized subtraction. We use a smooth color gradient from light blue to blue to represent increasing numbers of ON–OFF detectors, with $n=2$ shown in light blue and $n=\infty$ in blue. Solid, dashed, and dotted lines indicate detector efficiencies $\kappa = 1$, $0.95$, and $0.7$, respectively. All data are produced by  setting input states with 7 dB squeezing and an MPNR detector registering $k=2$ clicks.  
}
\label{fig:fid_prob_cat_breeding_k=2}
\end{figure*}



\subsection{Application to cat state generation}

Next, we consider the performance of MPNR detectors in the breeding process of Schrödinger cat states \cite{dakna1997generating,wenger2004pulsed,ourjoumtsev2006generating,lund2004conditional,takase2021generation,endo2023non,takase2023gottesman,takase2024generation} $|\rm cat\>_{\alpha}^\pm $, a protocol closely related to the detection of single-mode squeezed vacuum considered in the last subsection. Cat states play an important error-correcting role in quantum computing and communication. In particular, cat states can be used to further breed GKP states, as recently demonstrated experimentally~\cite{konno2024logical}. Here we focus on the even cat $|\rm cat\>_{\alpha}^+$, while the odd cat can be analyzed in a similar way.

As depicted in Fig.~\ref{fig:scheme_cat} (a), we first consider the generation of cat states via photon subtraction in a single-mode squeezed vacuum, with MPNR detection. The squeezed vacuum goes through a beamsplitter, diverting $1-\eta$ portion towards the MPNR detector while leaving $\eta$ portion leftover to approximate the cat state conditioned on the MPNR result. 
The following theorem gives a comparison between the fidelity achieved by using an ideal PNR detector and MPNR detectors. 

\begin{theorem} \label{lem:output state breeding cat}
In the process of generating cat states using a squeezed vacuum with a squeezing parameter $r$, a beamsplitter with transmissivity $\eta$, and an $n$-plexed PNR detector with a condition that the observed photon number is $k$, the resulting state can be expressed as follows: 
\begin{align}
\rho_{\widetilde{\rm cat}}
\propto &\frac{n!}{(n-k)!k!}\cdot \int \frac{\d^2 \alpha\d^2 \beta }{ \pi^2} \frac 1 {\cosh r}e^{-\tanh r \frac{\alpha^{*2}}{2 }}
e^{-\frac{|\alpha|^2}{2}} \nonumber \\
&\times e^{-\tanh r \frac{\beta^{*2}}{2 }}
e^{-\frac{|\beta|^2}{2}} \exp\left[-\frac{(1-\eta)(|\alpha|^2+|\beta|^2)}{2}\right] \nonumber \\
&\times \left\{\exp\left[\frac{(1-\eta)\alpha\beta^*}{n}\right]-1\right\}^k\left|\underline{\sqrt \eta \, i\alpha}\right\>\left\<\underline{\sqrt \eta \, i\beta}\right|.
\end{align}
Moreover, fidelity $F_{\rm mpnr}=\<{\rm cat}|_{\sqrt k}^+\rho_{\widetilde{\rm cat}}|{\rm cat}\>^+_{\sqrt k}$ is close to the fidelity, $F_{\rm pnr}= |\<{\rm cat}|^+_{\sqrt k}|\widetilde{\rm cat}\>_{\sqrt k}|^2$, achieved by an ideal PNR detector,
\begin{align}
F_{\rm mpnr}&=F_{\rm pnr}- \mathcal O\left(\frac{1-\eta}{n}\right) . 
\end{align}
\end{theorem}
A detailed proof of Theorem~\ref{lem:output state breeding cat} is presented in Appendix \ref{app:cat state breeding}. 

Note that fidelity alone cannot fully characterize probabilistic cat-state breeding because it trades off against the success probability. We therefore numerically simulate the protocol for the case with an input state with 7dB squeezing (see \cite{kashiwazaki2021fabrication,chen2022ultra}) and an MPNR detertor registering $k=2$ clicks  (see Fig.~\ref{fig:fid_prob_cat_breeding_k=2}). Here, we use a smooth color gradient from light blue to blue to show an improved maximal success probability at  fixed fidelity as more ON–OFF detectors are multiplexed. 
\QZ{Here, the success probability is obtained by computing the trace of the output state (see Eqs. (\ref{eqc15ppp}), (\ref{eqc23ppp}), and (\ref{eqc40ggg}) of Appendix \ref{app:cat state breeding}).}
Comparing solid, dashed, and dotted curves of the same color, which correspond to \QZ{ ON-OFF detector efficiencies} $\kappa = 1$, $\kappa = 0.95$, and $\kappa = 0.7$, respectively, the success probability is reduced accordingly.  
Notably, assuming an array of 20 detectors, we can achieve  fidelity $\sim0.88$ with success probability $\sim 3.8\%$, representing orders of magnitude higher success compared to previous works~\cite{endo2023non}. 


While the success probability is informative, in a real system, the state generation rate in unit of Hertz (Hz) is what matters the most. Next, we show that the parameters of the on-chip MPNR detector in Ref. \cite{cheng2023100} allows the generation of decent-quality cat states at megahertz rates. In comparison, earlier experiment based on a superconducting transition-edge sensor (TES) achieved only a 200Hz generation rate, constrained by a 20MHz detector count rate and an efficiency of $\kappa = 0.7$ \cite{endo2023non}. Although the TES resolved more than ten pixels, its modest efficiency forced the operation at high beamsplitter reflectivity to reach maximal fidelity (e.g., $\eta \approx 0.924$ for two-photon subtraction, $k=2$), which in turn limited the success probability in experiment to roughly $\sim0.004\%$ and kept the generation rate low. According to our simulations, replacing the TES with an ideal PNR detector while keeping efficiency, resolution, reflectivity, and effective squeezing level after loss 2.9dB as in Ref. \cite{endo2023non} yields no benefit. In this case, the success probability in simulation is approximately 0.02$\%$ and the fidelity is approximately 0.81 for two-photon subtraction.

In contrast, the spatial–temporal multiplexed PNR detector in Ref. \cite{cheng2023100} achieves high efficiency ($\kappa = 0.95$) with negligible dark counts. These advantages greatly increase the probability of success. Meanwhile, a ten-pixel configuration in~\cite{cheng2023100} can support a 100 MHz counting rate (though the original demonstration utilized 100 pixels with potentially lower counting rates), \QZ{computed by dividing the per-pixel rate (1 GHz~\cite{cheng2023100}) by the number of pixels.  } For two-photon subtraction with 7dB input squeezing, the highest fidelity, 0.84, appears at reflectivity $\eta\approx$0.7 and gives a 3.6$\%$ success probability. With 2.9dB squeezing, the fidelity peak moves to $\eta\approx$1; choosing $\eta$= 0.924 then yields 0.82 fidelity and a 0.03$\%$ success probability. Consequently, multiplying the 100MHz count rate by the corresponding success probabilities yields a megahertz-scale generation rate for 7dB squeezing and roughly 30kHz for 2.9dB squeezing.

We further examine generalized photon subtraction, where both beamsplitter inputs for cat-state breeding are squeezed vacuum states (see Fig. \ref{fig:scheme_cat}(b)). With $7\mathrm{dB}$ squeezing, a 20-pixel MPNR detector ($\kappa = 95\%$) registering $k = 2$ clicks achieves a fidelity $0.86$ with an $8.7\%$ success probability at beamsplitter reflectivity $\eta = 0.856$ (Fig.~\ref{fig:fid_prob_cat_breeding_k=2}).  
For $2.9\mathrm{dB}$ squeezing, the same setup reaches the same fidelity at $\eta = 0.803$, but with a reduced success probability of $1.7\%$. Hence, in both cases a megahertz-level generation rate is attainable with the MPNR detector of Ref.~\cite{cheng2023100}.

Finally, we extend our study to the four-photon subtraction protocol to compare with the results of~\cite{endo2025high}. Here, with 6.5dB input squeezing and beam-splitter reflectivity $\eta = 0.81$ fixed, we examine two detector settings. (i) An ideal PNR detector with efficiency $\kappa=0.95$  and unlimited pixels gives fidelity $0.88$, success probability $0.07\%$, and a generation rate of $3.5\text{kHz}$ at a 5MHz detector count rate. Limiting the MPNR to ten pixels lowers the fidelity to $0.73$ and the success probability to $0.04\%$, yielding a generation rate $2 \text{kHz}$. (ii) A TES with  efficiency $\kappa=0.4$ and unlimited pixels gives fidelity $0.58$, success probability $0.004\%$, and a generation rate $200\text{Hz}$. With ten pixels, the fidelity is $0.55$, the success probability $0.002\%$, and the generation rate $100\text{Hz}$. Even a ten-pixel MPNR already exceeds the rate $1.5\text{Hz}$ reported in Ref.~\cite{endo2025high}. A high-efficiency MPNR detector thus potentially boosts the rate by  orders of magnitude, showing that detector efficiency and pixel count now set the pace for high-rate, high-fidelity cat-state generation.

\QZ{In our simulations the parameter $\kappa$ denotes the overall per-pixel detection efficiency, including the intrinsic detector response and all coupling and filtering losses, and we consider $\kappa\in\{1,0.95,0.7,0.4\}$ to cover present and near-future devices. The value $\kappa=0.4$ used for the TES comparison is chosen as a conservative system efficiency consistent with recent telecom-band non-Gaussian state generation and four-photon subtraction experiments with TES-based PNR detectors \cite{endo2023non,endo2025high,sonoyama2023non}, while $\kappa\simeq0.7$ reflects the $\sim70\%$ internal efficiency of Ti--Au TESs at $1.5\,\mu\mathrm{m}$ and similar effective efficiencies in time-multiplexed or segmented PNR detectors \cite{endo2023non,fukuda2011titanium}. The near-unity value $\kappa=0.95$ is motivated by record system detection efficiencies of $95$–$98\%$ demonstrated in superconducting TES and SNSPD systems \cite{lita2008counting,fukuda2011titanium,dauler2014review}, and $\kappa=1$ is included only as an ideal reference without detection loss.
}

\section{ Conclusions and discussion}

In this work, we established a theoretical framework for photon counting with MPNR detectors. We demonstrate that, for estimating higher-order moments of the photon number, the error of MPNR scales inversely with the number of ON-OFF detectors $n$ employed in the MPNR detector, $\sim \mathcal O(1/n)$. Beyond the general analysis, we examine specific applications, including photon counting for squeezed vacuum states and the cat-state breeding process using MPNR detectors. We numerically illustrate the fidelity and success probability of the breeding process assisted by MPNR detectors.

Similar to dark count noise, we expect detector dead time to also introduce an error that scales as $1/n$, as long as the detectors are independent. The detector correlation is another important effect that may affect MPNR detectors.
In Appendix~\ref{app:correlated detectors}, 
we analyze the scenario in which each ON-OFF detector in the MPNR detector configuration exhibits correlations with an adjacent detector. Specifically, if one detector registers an ON event while its neighbor registers OFF, the correlated detection mechanism can still result in both detectors indicating ON with a probability $p$. Here we show that in this scenario, the estimation of the number of photons using an MPNR detector exhibits bias, which vanishes as the correlation probability $p$ approaches zero.\\

\QZ{
We also note that Refs.~\cite{sperling2012true} and~\cite{fiuravsek2025fundamental} work at the POVM level for arrays of ON–OFF detectors. Reference~\cite{sperling2012true} derives the exact click statistics of a given multiplexed array, including finite efficiency and noise, and uses this to analyze nonclassical features of light. Reference~\cite{fiuravsek2025fundamental} then exploits this POVM to obtain fundamental bounds on photon number probabilities, parity, and the mean photon number from tomographically incomplete click data for a fixed number of channels. In contrast, our framework is moment based and targets multiplexed photon-number-resolving detectors. For arbitrary finite-energy input states we show that all photon-number moments can be written in closed form in terms of Stirling numbers of the second kind, and we derive explicit $\mathcal O(1/n)$ convergence guarantees in the number 
$n$ of detector elements for any finite order moment. These analytic results are further combined with a concrete numerical scheme that simulates lossy cat-state breeding with multiplexed PNR detectors, which is not addressed in Refs.~\cite{sperling2012true,fiuravsek2025fundamental}.\\
}

\QZ{Beyond uniform dark-count noise and efficiency, detector dead time, and correlated-click errors, our framework still assumes an ideal linear optical network with almostly identical pixels and a single effective spatio-temporal mode. Imperfections such as mode mismatch, wavelength dependent splitting ratios, and readout bandwidth limits are not captured, and incorporating these device level constraints into the MPNR model is a natural direction for future work.
}

\section*{Acknowledgements}
 This work is supported by Defense Advanced Research Projects Agency (DARPA) HR00112490453. 

\section*{Data availability}
The data that support the findings of this article are openly available~\cite{data}.







\bibliographystyle{apsrev4-1}


%

\begin{widetext}

\appendix 

\section{Proof of Lemma \ref{theo:p_func_higher_moment} of the main text}\label{app:proof of lemma 2}

\subsection{Observables from ON-OFF detectors}

Without loss of generality, the MPNR detector can be defined as follows:
\begin{definition}[MPNR]An $n$-plexed PNR detector consists of an $n$-port beamsplitter that generates interference between the input state and vacuum ancillas. Each output port is connected to an ON-OFF detector described by the positive operator-valued measure  (POVM) $\{|0\>\<0|,I-|0\>\<0|\}$. 
\end{definition}

For clarity, we first focus on measuring photon distributions after the multiport beamsplitter distributing input states. Then, the following definition can be given:
\begin{definition}[Measurement after beamsplitting]\label{lem:povm:mPNR}
The POVM associated with the ON-OFF detectors at the output ports of the multiport beamsplitter in a MPNR detector is:
\begin{align}
&\left\{P_{{\rm ON-OFF},k},(k\le n)\right\}\nonumber \\
P_{{\rm ON-OFF},k}&=\bigoplus_{\sigma_g\in \map S_{n,k}}U_{\sigma_g} \left(\widetilde I^{\otimes k}\otimes |0\>\<0|^{\otimes n-k} \right) U_{\sigma_g}^\dag,\label{eq11}
\end{align}
where $U_{\sigma_g}$ denotes an $n$-mode permutation  operation, $\map S_{n,k}$ represents the permutation group for two sets of identical inputs, with $(n-k)$ and $k$ elements respectively, $\widetilde I=I-|0\>\<0|$ denotes a projector of a single mode that excludes the vacuum state $|0\>\<0|$. 
\end{definition}

Furthermore, the measurement outcomes obtained from the POVM defined in Definition  \ref{lem:povm:mPNR} can be used to estimate higher-order moments of the photon distribution through classical data processing. This relationship is formalized in the following lemma:
\begin{lemma}[Higher-order moments]\label{pro:higher_moment_observable}
Given a MPNR detector, the following observable can be realized: 
\begin{align}\label{eq8a}
\widehat N^h_{\rm ON-OFF}&=\left(\sum_{k=0}^\infty k\, P_{{\rm ON-OFF},k}\right)^h
\equiv \left(\sum_{j=1}^n \widetilde I_{j}\bigotimes_{\substack{k=1\\ k\neq j}}^n I_{k}\right)^h \nonumber \\
&(h= 1,2,\cdots ),
\end{align}
where $\widetilde I_{j}$ is a projector defined by $\widetilde I_{j}=I_{j}-|0\>\<0|_{j}:=\sum_{k_j=1}^\infty |k_j\>\<k_j|_{j}$, with $I_{j}$ referring to the identity operator of the $j$-th output port of the beamsplitter network, and $|0\>\<0|_{j}$ denoting the vacuum state of the $j$-th output mode. Intuitively, the definition of $\widehat N^h_{\rm ON-OFF}$ resembles the total photon number operator, but with the operator $\hat a_j^\dag \hat a_j$	 replaced by an alternative projection operator $\widetilde I_j$.  
\end{lemma}

\begin{proof}
The operator $\bigotimes_{\substack{k=1\\ k\neq j}}^n I_{k}$ can be expressed as a summation of terms involving products of $\widetilde I_k$ and $|0\>\<0|_k$. Explicitly, it includes terms such as $\widetilde I_1\otimes \widetilde I_2\cdots \widetilde I_{n}$, $|0\>\<0|_1\otimes \widetilde I_2\cdots \widetilde I_{n}$, and $|0\>\<0|_1\otimes |0\>\<0|_2\cdots |0\>\<0|_{n}$. Then, it is quick to prove that it can be rewritten as the summation of $k$ identical terms equal to $P_{{\rm ON-OFF},k}$. 
\end{proof}\\[-0.5em]

Using Definition \ref{lem:povm:mPNR} and Lemma \ref{pro:higher_moment_observable}, the distribution of total photon number can be determined. Then, one can combine them with the balanced beamsplitter operation $ U_{\rm BS}$ acted before ON-OFF detectors: 
\begin{align}
P_{{\rm ON-OFF},k}'&= U_{\rm BS}^\dag P_{{\rm ON-OFF},k}U_{\rm BS}\\
\widehat N^{'h}_{\rm ON-OFF}&= U_{\rm BS}^\dag\widehat N^h_{\rm ON-OFF}U_{\rm BS}\label{eqa4aa}
\end{align}
where $U_{\rm BS}$ is a beamsplitter operation that satisfies the relation $U_{\rm BS}\hat a_k U_{\rm BS}^\dag= \frac 1 {\sqrt n}\sum_{j=1}^n \,\omega_{jk}^{1/2}\hat a_j,\ k=1\cdots,n$ with the weights $\{\omega_{jk}\}$ fulfilling the normalization condition $\sum_{j=1}^n\left(\omega_{jk}\omega_{k\ell}\right)^{1/2}=\delta_{j\ell}$. 
Without losing the generality, all ancillas are assumed to be vacuum states. Therefore, this operation allows for a complete characterization of the measurement process in a MPNR detector. 

In the following subsections, we will present an application to the detection of photon distribution and its higher moments.

\subsection{Effective realization of PNR with ON-OFF detectors}\label{app:effective implementationo f PNR with ON-OFF}


\subsubsection{Estimation of mixed states}

Let's look at the performance of MPNR detectors in estimation photon moments. On the side of probe state, we can consider an arbitrary single-mode quantum state expressed with P-function: 
\begin{align}\label{p_function}
\rho&=\int \frac{\d^2 \alpha}{\pi} P(\alpha) |\underline{\alpha}\>\<\underline{\alpha}|.  
\end{align}
where $|\underline{\alpha}\>$ with $\alpha\in \mathbb C$ denotes the coherent state. The beamsplitter network generating interference between the input and vacuum states will transform the state $\rho$ into the following form: 
\begin{align}
\rho'=& \int \frac{\d^{2}  \alpha}{\pi} P(\alpha)\bigotimes_{j=1}^n \left|\underline{\sqrt{\eta_j'}{\alpha}}\right\>\left\<\underline{\sqrt{\eta_j'}{\alpha}}\right| ,
\end{align}
where $\{\eta_j''=\omega_{1j}\}$ denotes the reflectivity, with $\omega_{1j}$ being defined in Eq. (\ref{eqa4aa}). In addition, each ON-OFF detector has an efficiency $\kappa_j$. Then, the state becomes:   
\begin{align}\label{eqb6}
\rho''=& \int \frac{\d^{2}  \alpha}{\pi} P(\alpha)\bigotimes_{j=1}^n \left|\underline{\sqrt{\eta_j'\kappa_j}{\alpha}}\right\>\left\<\underline{\sqrt{\eta_j'\kappa_j}{\alpha}}\right| .
\end{align}

On the side of measurements, the ON-OFF detectors can only lead to an observable $\widehat N^h_{\rm ON-OFF}$ \QZ{defined in Eq. (\ref{eq8a}),  }
\QZ{ which just replaces the energy operator $\{\hat a^\dag_j \hat a_j\}$ with an alternative projector $\widetilde I_{j}$.} 

\QZ{
Without losing the generality, let's consider the scenario with presumptions 
$\eta''_j=\eta_j'\kappa_j\simeq \eta''=\mathcal O(1/n)$ and $n\gg h=\mathcal O(1)$. Then, applying ON-OFF detectors to estimate the $h$-th moment of photon number ($h\le n$) will lead to the following result: 
\begin{align}
N^h_{\rm mpnr}=&\left\<\left(\sum_{i=1}^n \widetilde I_i \bigotimes_{\substack{j=1\\ j\neq i}}^n I_j \right)^h\right\>_{\rho''}\\
=&\int \frac{\d^2 \alpha}{\pi} P(\alpha) \sum_{m=1}^h 
\binom{n}{m}\sum_{k=0}^m (-1)^k \binom{m}{k}(m-k)^h\left\{\left[1-\exp\left(-\frac {|\alpha|^2} n \right)\right]^m \right.\nonumber \\
&\left. +\mathcal O\left(\max_j\left|n\eta_j''-1 \right|\cdot |\alpha|^{2}\right)\right\}\label{eqa9aaa} \\
=&\int \frac{\d^2 \alpha}{\pi} P(\alpha) \sum_{m=1}^h \binom{n}{m}\sum_{k=0}^m (-1)^k \binom{m}{k}(m-k)^h\left(\frac{|\alpha|^2}n\right)^m\left[1-\frac {m|\alpha|^2}{2n}+\mathcal O\left(\frac{m^2|\alpha|^4}{n^2}\right)\right.\nonumber\\
&\left.\ \ \ \ \ +\mathcal O\left(\max_j\left|\eta_j''-\frac 1 n \right|\right)\right] \\
=&\int \frac{\d^2 \alpha}{\pi} P(\alpha)\sum_{m=1}^h \frac{1}{m!}\left[1-\frac{m(m-1)}{2n}+\mathcal O\left(\frac{1}{n^2}\right)\right] \sum_{k=0}^m (-1)^k \binom{m}{k}(m-k)^h |\alpha|^{2m} \nonumber \\
&-\int \frac{\d^2 \alpha}{\pi} P(\alpha) \sum_{m=1}^h \binom{n}{m}\sum_{k=0}^m (-1)^k \binom{m}{k}(m-k)^h\left(\frac{|\alpha|^2}n\right)^m\left[\frac {m|\alpha|^2}{2n}+\mathcal O\left(\frac{|\alpha|^4}{n^2}\right)\right.\nonumber \\
&\ \ \ \ \ \left.+\mathcal O\left(\max_j\left|\eta_j''-\frac 1 n\right|\right)\right] \label{eqa11fff}\\
=&N^h - \mathcal O \left(\frac {h^2} n\right)+\mathcal O\left(\max_j\left|\eta_j''-\frac 1 n\right|\right)\label{eqa12ggg}\\
=&N^h +\mathcal O\left(\frac{1}{n}\right)\label{eqc13fff}
\end{align}
where Eq.~(\ref{eqa9aaa}) follows by expanding $\Bigg(\sum_{i=1}^n \widetilde I_i \bigotimes_{\substack{j=1\\ j\neq i}}^n I_j \Bigg)^{\!h}$ and applying each term to the coherent product state \(\bigotimes_{j=1}^n \left|\underline{\sqrt{\eta_j''}\alpha}\right\>\) introduced in Eq. (\ref{eqb6}) with an assumption $\eta''=\mathcal O(1/n)$. Here, the number of terms with exactly \(m\) occurrences of \(\widetilde I\) equals
\begin{align}
\binom{n}{m}\sum_{k=0}^m (-1)^k \binom{m}{k}(m-k)^h
\end{align}
which is obtained by first choosing $m$ detecting modes among the $n$ modes, giving $\binom{n}{m}$ options, and then counting the surjective maps from the $h$ positions to these $m$ modes, whose number is $m!\,S(h,m)=\sum_{k=0}^m (-1)^k \binom{m}{k}(m-k)^h$, where $S(h,m)$ denotes the Stirling number of the second kind. Then, each term with $m$ occurrences of $\widetilde I$ projects the coherent product state \(\bigotimes_{j=1}^n \bigl|\underline{\sqrt{\eta_j''}\alpha}\bigr\rangle\) into the same scalar factor
\begin{align}
&\Bigg[1-\exp\!\left(-\tfrac{|\alpha|^2}{n}-\mathcal O\left(\max_j\left|\eta_j''-\frac 1 n\right|\right)|\alpha|^2\right)\Bigg]^m
\\
=&\left[1-\exp\left(-\frac {|\alpha|^2}{n} \right)\right]^m 
+\mathcal O\left(\max_j\left|n\eta_j''-1\right|\cdot |\alpha|^{2}\right)
\end{align}
}

\QZ{
Further, Eq. (\ref{eqa11fff}) is obtained by the property: 
\begin{align}
\binom{n}{m}\frac 1 {n^m}&=\frac{n(n-1)(n-2)\cdots (n-m+1)}{m!n^m}\\
&=\frac{1}{m!}\left[1-\frac{m(m-1)}{2n}+\mathcal O\left(\frac{1}{n^2}\right)\right]
\end{align}
}

\QZ{
Eq. (\ref{eqa12ggg}) follows from the relation as follows: 
\begin{align}
N^h=&\left\<\left(\sum_{i=1}^n \hat a_i ^\dag \hat a_i \otimes_{j=1\atop j\neq i}^n I_j\right)^h\right\>_{\rho''}\\
=&\int \frac{\d^2 \alpha}{\pi} P(\alpha)\left\<\underline{\frac{\alpha}{\sqrt n}}\right|^{\otimes n}\left(\sum_{i=1}^n \hat a_i ^\dag \hat a_i \otimes_{j=1\atop j\neq i}^n I_j\right)^h\left|\underline{\frac{\alpha}{\sqrt n}}\right\>^{\otimes n} \\
=&\int \frac{\d^2 \alpha}{\pi} P(\alpha)\left\<\underline{\alpha}\right|\left\<\underline{0}\right|^{\otimes (n-1)}\left(\sum_{i=1}^n \hat a_i ^\dag \hat a_i \otimes_{j=1\atop j\neq i}^n I_j\right)^h\left|\underline{\alpha}\right\>\left|\underline{0}\right\>^{\otimes (n-1)} \label{eqa21iii}\\
=&\int \frac{\d^2 \alpha}{\pi} P(\alpha)\left\<\underline{\alpha}\right|\left( \hat a_1^\dag \hat a_1 \right)^h\left|\underline{\alpha}\right\>\\
=& \int \frac{\d^2 \alpha}{\pi} P(\alpha)\sum_{m=1}^h\frac 1 {m!}\sum_{k=0}^m (-1)^k \binom{m}{k}(m-k)^h |\alpha|^{2m}\label{eqa23ppp}
\end{align}
where Eq. (\ref{eqa21iii}) is derived using the fact that the total photon number operator $\sum_{i=1}^n \hat a_i ^\dag \hat a_i \otimes_{j=1\atop j\neq i}^n I_j$ commutes with the beamsplitter operation that concentrates all the photons into the first mode, Eq. (\ref{eqa23ppp}) follows from Eq. (1.37) of \cite{mansour2016commutation}.
}

By using Eq.  (\ref{eqc13fff}), one have the conclusion in Lemma \ref{theo:p_func_higher_moment} of the main text.


\subsubsection{Estimation of pure states}

Similar to Eq. (\ref{p_function}), an arbitrary pure state $|\psi\>$ can be decomposed as follows:  
\begin{align}
|\psi\>&=\int \frac{\d^2 \alpha}{\pi} |\underline{\alpha}\>\<\underline{\alpha}|\psi\>\\
&=\int \frac{\d^2 \alpha}{\pi} f(\alpha) |\alpha\>\label{p_function2}
\end{align}
where the first equality is derived from the super completeness of coherent states, $f(\alpha):=\<\underline{\alpha}|\psi\>$ is a generalized P-function that defines the state $|\psi\>$. A beamsplitter network interacting with vacuums will generate a state: 
\begin{align}\label{eqc3}
|\psi'\>&=\int \frac{\d^2 \alpha}{\pi} f(\alpha) \bigotimes_{j=1}^n \left|\underline{\sqrt{\eta_j''}\alpha}\right\>. 
\end{align}

Given an arbitrary pure state defined in Eq. (\ref{p_function2}), a  beamsplitter lead to a state $|\psi'\>=\int \frac{\d^2 \alpha}{\pi} f(\alpha)\left|\underline{\alpha/\sqrt n}\right\>^{\otimes n}$. Then, the measurement result can be explicitly written as follows: 
\QZ{
\begin{align}
N^h_{\rm mpnr}=&\int \frac{\d^2 \alpha\d^2 \alpha'}{\pi^2} f(\alpha)f(\alpha') \sum_{m=1}^h 
\binom{n}{m}\sum_{k=0}^m (-1)^k \binom{m}{k}(m-k)^h\nonumber \\&\times \left[\left\<\underline{\frac{\alpha}{\sqrt n}}\right|\left.\underline{\frac{\alpha'}{\sqrt n}}\right\>-\exp\left(-\frac{|\alpha|^2+|\alpha'|^2}{2n}\right)\right]^m \left\<\underline{\frac{\alpha}{\sqrt n}}\right|\left.\underline{\frac{\alpha'}{\sqrt n}}\right\>^{n-m} \label{eqb10},
\end{align}
where $\<\underline{\beta}|\underline{\alpha}\>=\exp\left[-\frac 1 2 \left(|\beta|^2+|\alpha|^2-2\beta^*\alpha\right)\right]$ denotes the overlap between two coherent states $|\alpha\>$ and $|\beta\>$. This expression will be used for numerical simulation in the next subsection. 
}

\QZ{
Note that the coherent state basis is over complete, Eq. (\ref{eqc13fff})  works for all pure states, i.e., $N^h_{\rm mpnr}=N^h +\mathcal O\left(\frac{1}{n}\right)$. 
}

\subsubsection{Numerical evaluation for typical states}

By choosing the p-function in Eq. (\ref{eqb6}) as a Dirac delta, one can immediately obtain the explicit photon number moment $N^h_{\rm ON-OFF}$ for coherent states. In Fig. \ref{fig:cohernt_state} of the main text, we examine the error $| N^2_{\rm mpnr}-N^2|$ versus the true value of $N^2$. It can be observed that the scaling of quantity $| N^2_{\rm mpnr}-N^2| /N^2$ converges to $\mathcal O(\frac 1 n)$ as the number of detectors increases. Similarly, we can apply Eq. (\ref{eqb10}) to a Schr\"odinger cat state $|\text{cat}\>\propto |\underline{\beta}\>+|\underline{-\beta}\>$ with a distribution $f(\alpha)=1/\sqrt{2+2e^{-2|\beta|^2}}\delta_{\alpha-\beta}+1/\sqrt{2+2e^{-2|\beta|^2}}\delta_{\alpha+\beta}$. An evaluation of the performance is also shown in Fig. \ref{fig:cohernt_state} of the main text.

\section{Measuring photon distribution of squeezed vacuum states through MPNR detectors}\label{app:calibration_squeezied state}

Without losing the generality, let's consider a balanced beamsplitter and a squeezed state $|r\>_{\rm sq}$ with a real squeezed parameter $z=r\in \mathbb R$ and $\phi=\pi$. Given by the super-completeness of coherent states $\int \frac{\d^2\alpha}{\pi}|\underline \alpha\>\<\underline \alpha|=I$, we have the following expression: 
\begin{align}
|r\>_{\rm sq}&=\int \frac{\d^2\alpha}{\pi}|\underline\alpha\>\<\underline\alpha|\cdot  \frac{1}{\sqrt{\cosh r}} \sum_{n=0}^\infty (-\tanh r)^n \frac{\sqrt{(2n)!}}{2^n n!}|2n\>\\
&= \int \frac{\d^2\alpha}{\pi} \frac{1}{\sqrt{\cosh r}} \sum_{n=0}^\infty (-\tanh r)^n \frac{\alpha^{*2n}}{2^n n!}
e^{-\frac{|\alpha|^2}{2}} |\underline \alpha\>\\
&= \int \frac{\d^2\alpha}{\pi} \frac{1}{\sqrt{\cosh r}} e^{-\tanh r \frac{\alpha^{*2}}{2 }}
e^{-\frac{|\alpha|^2}{2}} |\underline\alpha\>\label{eqc8}
\end{align}

Then, consider an $n$-port MPNR detector with a balanced beamsplitter. Using Eq. (\ref{eqc8}), we have the state to be measured: 
\begin{align}
|r\>'_{\rm sp}&= \int \frac{\d^2\alpha}{\pi} \frac{1}{\sqrt{\cosh r}} e^{-\tanh r \frac{\alpha^{*2}}{2 }}
e^{-\frac{|\alpha|^2}{2}} \left|\underline{\frac{\alpha}{\sqrt n}}\right\>^{\otimes n}\\
&= \sum_{\ell=0}^\infty \sum_{j_1+j_2+\cdots=2\ell\atop j_1,j_2\cdots=0,\cdots,2\ell}  \int  \d \xi\, \frac{2\xi}{\sqrt{\cosh r}} e^{-\xi^2}\frac{(-\xi^2\tanh r)^\ell  }{2^\ell\ell!} \frac{\xi^{2\ell}}{n^\ell\sqrt{j_1!j_2!,\cdots}}|j_1\>|j_2\>\cdots |j_n\>\\
&= \sum_{\ell=0}^\infty \sum_{j_1+j_2+\cdots=2\ell\atop j_1,j_2\cdots=0,\cdots,2\ell}  \frac{1}{\sqrt{\cosh r}} \frac{(-\tanh r)^\ell  }{2^\ell\ell!} \frac{(2\ell)!}{n^\ell\sqrt{j_1!j_2!,\cdots}} |j_1\>|j_2\>\cdots |j_n\>\\
&= \sum_{\ell=0}^\infty \frac{1}{\sqrt{\cosh r}} \frac{(-\tanh r)^\ell \sqrt{(2\ell)!} }{2^\ell\ell!} |\Phi_\ell\>
\end{align}
where $\alpha=\xi e^{i\phi}$ refer to the reparametrization of displacement, $|\Phi_\ell\>$ is defined as follows: 
\begin{align}
|\Phi_\ell\>&=  \sum_{j_1+j_2+\cdots=2\ell\atop j_1,j_2\cdots=0,\cdots,2\ell}  \frac{\sqrt{(2\ell)!}}{n^{\ell}\sqrt{j_1!j_2!,\cdots}} |j_1\>|j_2\>\cdots |j_n\>
\end{align}

Then, the observed photon number distribution of a squeezed vacuum state $|z\>_{\rm sq}$ is: 
\begin{align}
\widetilde p_{{\rm sq},k}&=  \<r|'_{\rm sq}P_{{\rm ON-OFF},k}|r\>'_{\rm sq}\\
&= \<r|'_{\rm sq}\bigoplus_{\sigma_g\in \map S_{n,k}}U_{\sigma_g} \left(\widetilde I^{\otimes k}\otimes |0\>\<0|^{\otimes n-k} \right) U_{\sigma_g}^\dag |r\>'_{\rm sq} \\
&= \frac{n!}{k!(n-k)!} \<r|'\sum_{\ell=0}^\infty \sum_{j_{1}+\cdots+j_k=2\ell\atop j_{1},\cdots,j_k=1,\cdots,2\ell}  \frac{1}{\sqrt{\cosh r}} \frac{(-\tanh r)^\ell  }{2^\ell\ell!} \frac{(2\ell)!}{n^\ell\sqrt{j_1!j_2!,\cdots}} |j_{1}\>\cdots  |j_k\>|0\>^{\otimes n-k}\\
&= \frac{n!}{k!(n-k)!} \sum_{\ell=0}^\infty \frac{1}{\cosh r} \frac{(\tanh r)^{2\ell}  }{2^{2\ell}(\ell!)^2} \frac{[(2\ell)!]^2}{n^{2\ell}} \sum_{j_{1}+\cdots+j_k=2\ell\atop j_{1},\cdots,j_k=1,\cdots,2\ell}  \frac{1}{j_{1}!\cdots j_k!} 
\end{align}
where $U_{\sigma_g}$ denotes an $n$-mode permutation operation, $\map S_{n,k}$ represents the permutation group for two sets of identical inputs, with $(n-k)$ and $k$ elements respectively, $\widetilde I=I-|0\>\<0|$ denotes a projector of a single mode that excludes the vacuum state $|0\>\<0|$. 

\QZ{
Given the relation as follows: 
\begin{align}
\left(e^x-1\right)^k&= \left(\sum_{j=1}^\infty \frac{x^j}{j!}\right)^k\\
&= \sum_{n'\ge k}\left(\sum_{j_1+j_2+\cdots +j_k=n'\atop j_1,\cdots,j_k= 1,\cdots,n'}\frac 1 {j_1! \cdots j_k!}\right) x^{n'}  \label{eqb14ppp}\\
&\equiv k! \sum_{n'\ge k } \frac{x^{n'} }{n'!} \sum_{i=0}^k (-1)^i\left(\begin{matrix}
k\\i
\end{matrix}\right)(k-i)^{n'}\label{eqb15ppp}
\end{align}
where the last equation follows from Eq. (6.3) of \cite{boyadzhiev2012close}. Comparing the coefficient associated with the $x^{2\ell}$ term, we have :
\begin{align}
\sum_{j_1+j_2+\cdots +j_k=2\ell\atop j_1,\cdots,j_k= 1,\cdots,2\ell}\frac 1 {j_1! \cdots j_k!}&=  \frac{k!\, }{(2\ell)!} \sum_{i=0}^k (-1)^i\left(\begin{matrix}
k\\i
\end{matrix}\right)(k-i)^{2\ell}\\
&= \frac{k!\, }{(2\ell)!} S(2\ell,k)
\end{align}
where $S(h,k)$ denotes the Stirling numbers of the second kind. Thus, the probability $\widetilde p_{{\rm sq},k}$ can be written in an alternative way: 
\begin{align}
\widetilde p_{{\rm sq},k}
&= \frac{n!}{(n-k)!} \sum_{\ell=0}^\infty \frac{1}{\cosh r} \frac{(\tanh r)^{2\ell}  }{2^{2\ell}(\ell!)^2} \frac{(2\ell)!}{n^{2\ell}} S(2\ell,k).
\end{align}
}

On the other hand, the true value of photon number distribution is: 
\begin{align}
p_{{\rm sq},k}&= \begin{cases}
\frac{1}{\cosh r} \frac{(\tanh r)^{2k} (2k)! }{2^{2k}(k!)^2} & k \text{\ is even}\\
0 & k \text{\ is odd}\\
\end{cases}.
\end{align}


\section{Application to cat state generation protocols}\label{app:cat state breeding}

It has been widely shown that a significant challenge in approximating GKP states lies in the implementation of PNR detectors when preparing cat states \cite{takase2021generation,endo2023non,takase2023gottesman,takase2024generation}, which typically serves as the initial step in generating GKP states \cite{vasconcelos2010all,weigand2018generating,konno2024logical}. The first scheme for generating cat states was proposed in \cite{dakna1997generating}, based on a breeding process using a squeezed input state and PNR detection. Later, the first experimental realization of approximate small-amplitude cat states, conditioned by single-photon detection, was reported in \cite{wenger2004pulsed} and \cite{ourjoumtsev2006generating}, following the theoretical proposal in \cite{lund2004conditional}. Additionally, a method for breeding large-amplitude cat states from small-amplitude ones was proposed in \cite{lund2004conditional}. Subsequent theoretical and experimental work demonstrated that the roles of the squeezed state and PNR detector could be interchanged by breeding from Fock states with homodyne detection \cite{ourjoumtsev2007generation}. However, the realization of a ``true'' PNR detector and deterministic generation of Fock states remains challenging tasks \cite{deng2024quantum}. The generation of cat states with higher probabilities remains an open problem in the preparation of logical states and error correction \cite{de2022error,campagne2020quantum,eickbusch2022fast}. In this section, we investigate the performance of using ON-OFF detectors in cat state breeding.

\subsection{Explicit expression of Schr\"odinger cat states }

The Schr\"odinger cat states are defined as follows: 
\begin{align}
|{\rm cat}\>^+_\alpha=&\map N^{-1}_\alpha \left(|\underline{\alpha}\>+|\underline{-\alpha}\>\right)    \\
=& \map N^{-1}_\alpha e^{-\frac{|\alpha|^2}{2}}\sum_{j=0}^\infty \frac{2\alpha^{2j}}{\sqrt{(2j)!}}|2j\>\label{eqd2}\\
|{\rm cat}\>^{-}_\alpha=&\map N^{-1}_\alpha \left(|\underline{\alpha}\>-|\underline{-\alpha}\>\right)    \\
=& \map N^{-1}_\alpha e^{-\frac{|\alpha|^2}{2}}\sum_{j=0}^\infty \frac{2\alpha^{2j+1}}{\sqrt{(2j+1)!}}|2j+1\>\label{eqd4}
\end{align}
where $\map N_{\alpha}=\sqrt{2+2e^{-2|\alpha|^2}}$ is a constant for normalization. 

\subsection{Breeding via photon subtraction}

\subsubsection{Breeding cat states from squeezed states}

In practice, the Schr\"odinger cat state $|{\rm cat}\>_\alpha\propto |\underline{\alpha}\>+|\underline{-\alpha}\>$ is usually generated in the following process \cite{dakna1997generating,takase2023gottesman}: 
\begin{align}
|{\rm cat}\>_\alpha\approx \<k| V(|z\>_{\rm sq}\otimes |0\>),
\end{align}
where $V$ denotes a beamsplitter operation: $VaV^\dag = \sqrt {1-\eta} a-i\sqrt{\eta}b$, $VbV^\dag = -i\sqrt{\eta} a+\sqrt{1-\eta}b$ with $a(b)$ being the annihilation operator of the first (second) mode, $\eta$ being the reflectivity satisfying $0\le \eta \le 1$, $|z\>_{\rm sq}=\frac 1 {\sqrt{\cosh r }} \sum_{j=0}^\infty  (-e^{i\phi}\tanh r )^j  \frac{\sqrt{(2j)!}}{2^j j!}|2j\>$ is a single-mode squeezed vacuum state, $\<k|$ denotes a probabilistic conditioning that projecting the input onto a Fock state $|k\>,k\in \mathbb N$.

Using the result in Eq. (\ref{eqc8}), the explicit expression of the output state before photon detection can be derived:
\begin{align}\label{eqc6}
|\Psi\>= & \int \frac{\d^2 \alpha}{\pi} \frac{1}{\sqrt{\cosh r}} e^{-\tanh r \frac{\alpha^{*2}}{2 }}
e^{-\frac{|\alpha|^2}{2}} \left|\underline{\sqrt{1-\eta }\,\alpha}\right\>\left|\underline{\sqrt \eta \, i\alpha}\right\>. 
\end{align}

Then, one can probabilistically generate a cat state by projecting the second mode into a number state $|k\>,k\in \mathbb N$. The explicit expression for the resulting state is:

\begin{align}
|\widetilde{\rm cat}\>_{\sqrt k}\propto  & \int \frac{\d^2 \alpha}{\pi } \frac{1}{\sqrt{\cosh r}} e^{-\tanh r \frac{\alpha^{*2}}{2 }}
e^{-\frac{|\alpha|^2}{2}} e^{-\frac{(1-\eta)|\alpha|^2}{2}}\frac{(1-\eta)^{\frac k 2 }\alpha^{k}}{\sqrt{k!}} \left|\underline{\sqrt \eta\, i\alpha}\right\>\\
=& \sum_{j=0}^\infty \int \frac{\d^2 \alpha}{ \pi}\frac{1}{\sqrt{\cosh r}} e^{-\tanh r \frac{\alpha^{*2}}{2 }}
e^{-|\alpha|^2}  \frac{(1-\eta)^{\frac k 2 }\eta^{\frac j 2 }i^j\alpha^{k+j}}{\sqrt{k!j!}} \left|j\right\>\label{eqd11}\\
=& 
\begin{cases}
&\sum_{j=0}^\infty \int \frac{\d^2 \alpha}{ \pi}\frac{1}{\sqrt{\cosh r}} \left(-\tanh r \frac{\alpha^{*2}}{2 }\right)^{(k+2j)/2}\frac{1}{(k/2+j)!}
e^{-|\alpha|^2}  \frac{(1-\eta)^{\frac k 2 }\eta^{j  }i^{2j}\alpha^{k+2j}}{\sqrt{k!(2j)!}} \left|2j\right\>,\ \ \ \ \ \ \ \ \ \ \ \ (k{\rm\ is\ even})\\
&\sum_{j=0}^\infty \int \frac{\d^2 \alpha}{ \pi}\frac{1}{\sqrt{\cosh r}} \left(-\tanh r \frac{\alpha^{*2}}{2 }\right)^{(k+2j+1)/2}\frac 1 {((k+2j+1)/2)!}
e^{-|\alpha|^2 } \frac{(1-\eta)^{\frac k 2 }\eta^{j+1/2  }i^{2j+1}\alpha^{k+2j+1}}{\sqrt{k!(2j+1)!}} \left|2j+1\right\>,\\ &(k{\rm\ is\ odd})
\end{cases}\label{eqd5}\\
=& 
\begin{cases}
&\sum_{j=0}^\infty  \frac{1}{\sqrt{\cosh r}} \left(-\frac{1}{2 }\tanh r \right)^{(k+2j)/2}\frac{(k+2j)!}{(k/2+j)!}
\frac{(1-\eta)^{\frac k 2 }\eta^{j  }i^{2j}}{\sqrt{k!(2j)!}} \left|2j\right\>,\ \ \ \ \ \ \ \ \ \ \ \ (k{\rm\ is\ even})\\
&\sum_{j=0}^\infty \frac{1}{\sqrt{\cosh r}} \left(-\frac{1}{2 }\tanh r \right)^{(k+2j+1)/2}\frac {(k+2j+1)!} {((k+2j+1)/2)!}
\frac{(1-\eta)^{\frac k 2 }\eta^{j+1/2  }i^{2j+1}}{\sqrt{k!(2j+1)!}} \left|2j+1\right\>,\ \ \ (k{\rm\ is\ odd})
\end{cases}\label{eqc13}
\end{align}
where Eq. (\ref{eqd5}) is obtained from the rotational symmetry of the amplitude to be integrated.

\subsubsection{Application of MPNR detectors}

As shown in Section \ref{app:effective implementationo f PNR with ON-OFF}, PNR detectors can be effectively implemented using ON-OFF detectors with assistance of a beamsplitter network. For an ideal PNR detector, if we probabilistically leave the state unchanged when the outcome is $k$ and discard the output with other outcomes, the state will be projected onto a Fock state $|n\>$. Now, let's consider an $n$-port balanced  beamsplitter with each output port being measured with an ON-OFF detector. The goal is to estimate the photon number $N$. When the measurement outcome is $k$, it is associated with the following projector: 
\begin{align}
\map P_{{\rm ON-OFF},k}&=\bigoplus_{\sigma_g\in \map S_{n,k}}U_{\sigma_g} \left(|0\>\<0|^{\otimes n-k}\otimes \widetilde I^{\otimes k}\right) U_{\sigma_g}^\dag,\ \ \ \ (k\le n)   \label{eqc10}
\end{align}
where $U_{\sigma_g}$ denotes the permutation unitary operation, $\map S_{n,k}$ represents the permutation group for two sets of identical inputs, with $(n-k)$ and $k$ elements respectively, $\widetilde I_j=I_j-|0\>\<0|_j$ denotes a projector of the $j$-th mode that excludes the corresponding vacuum state $|0\>\<0|_j$. 

Now, let's look at the breeding process. A beamsplitter will generate interference between a single-mode squeezed state and a vacuum state. The output state can be written in Eq. (\ref{eqc6}). Then, one can implement an $n$-port balanced beamsplitter acting on the first mode of the state in Eq. (\ref{eqc6}) and to achieve the following state: 
\begin{align}
|\Psi'\>= & \int \frac{\d^2 \alpha}{\pi} \frac{1}{\sqrt{\cosh r}} e^{-\tanh r \frac{\alpha^{*2}}{2 }}
e^{-\frac{|\alpha|^2}{2}} \left|\underline{\frac{\sqrt{1-\eta }\,\alpha}{\sqrt n }}\right\>^{\otimes n}\left|\underline{\sqrt \eta \, i\alpha}\right\>.\label{eqc11}
\end{align}
Further, the photon number measurement with ON-OFF detectors gives an outcome $k$. The whole process will project the overall state as follows: 
\begin{align}
\rho_{\widetilde{\rm cat}}\propto &\Tr \left\{ \left\{ \left[\bigoplus_{\sigma_g\in \map S_{n,k}}U_{\sigma_g} \left(|0\>\<0|^{\otimes n-k}\otimes \widetilde I^{\otimes k}\right) U_{\sigma_g}^\dag\right] \otimes I \right\}|\Psi'\>\<\Psi'| \right\}\\
= &\frac{n!}{(n-k)!k!}\cdot \Tr \left[\left( |0\>\<0|^{\otimes n-k}\otimes \widetilde I^{\otimes k} \otimes I  \right)|\Psi'\>\<\Psi'|\right]\label{eqc13a}\\ 
= &\frac{n!}{(n-k)!k!}\cdot \int \frac{\d^2 \alpha\d^2 \beta }{ \pi^2} \frac 1 {\cosh r}e^{-\tanh r \frac{\alpha^{*2}}{2 }}
e^{-\frac{|\alpha|^2}{2}} e^{-\tanh r \frac{\beta^{*2}}{2 }}
e^{-\frac{|\beta|^2}{2}} \nonumber \\
&\times \exp\left[-\frac{(1-\eta)(|\alpha|^2+|\beta|^2)}{2}\right] \left\{\exp\left[\frac{(1-\eta)\alpha\beta^*}{n}\right]-1\right\}^k\nonumber \\
&\times \left|\underline{\sqrt \eta \, i\alpha}\right\>\left\<\underline{\sqrt \eta \, i\beta}\right|\label{eqc15ppp}
\end{align}
where Eq. (\ref{eqc13a}) is derived using the fact that the state $|\Psi'\>$ is invariant under any permutation operation on the first $n$ modes.

Without losing the generality, let's make the following assumption: the number of ON-OFF detectors $n$ is much larger that of the average photon number of the input state. On this account, we will have the following approximation: 
\begin{align}
\left.\rho_{\widetilde{\rm cat}}\right|_{1-\eta=\mathcal O(n^{-1})\atop n\gg \sinh^2 r}\propto &\frac{n!}{(n-k)!k!}\cdot \int \frac{\d^2 \alpha\d^2 \beta }{ \pi^2}\frac 1 {\cosh r}e^{-\tanh r \frac{\alpha^{*2}}{2 }}
e^{-\frac{|\alpha|^2}{2}} e^{-\tanh r \frac{\beta^{*2}}{2 }}
e^{-\frac{|\beta|^2}{2}} \nonumber \\
&\times \exp\left[-\frac{(1-\eta)(|\alpha|^2+|\beta|^2)}{2}\right] \left[\frac{(1-\eta)\alpha\beta^*}{n}\right]^k\left[1+\frac{k}{2}\frac{(1-\eta)\alpha\beta^*}{n}+\mathcal O\left(\frac{k^2(1-\eta)^2(\alpha\beta^*)^2}{n^2}\right)\right]\nonumber \\
&\times \left|\underline{\sqrt \eta \, i\alpha}\right\>\left\<\underline{\sqrt \eta \, i\beta}\right|\\
=&\frac{n!}{(n-k)!k!}\cdot \int \frac{\d^2 \alpha\d^2 \beta }{ \pi^2}\frac 1 {\cosh r}e^{-\tanh r \frac{\alpha^{*2}}{2 }}
e^{-|\alpha|^2} e^{-\tanh r \frac{\beta^{*2}}{2 }}
e^{-|\beta|^2}\nonumber \\
&\times \left[\frac{(1-\eta)\alpha\beta^*}{n}\right]^k\left[1+\frac{k}{2}\frac{(1-\eta)\alpha\beta^*}{n}+\mathcal O\left(\frac{k^2(1-\eta)^2(\alpha\beta^*)^2}{n^2}\right)\right]\nonumber \\
&\times \sum_{j,\ell=0}^\infty \frac{\eta^{(j+\ell)/2}i^j(-i)^\ell\alpha^j\beta^{*\ell }}{\sqrt{j!\ell!}}|j\>\<\ell|.\label{eqd20}
\end{align}

Comparing the state in Eq. (\ref{eqd20}) with Eq. (\ref{eqd11}), one have the following relation: 
\begin{align}
F_{\rm mpnr}&=F_{\rm pnr}- \mathcal O\left(\frac{1-\eta}{n}\right) . 
\end{align}
where $F_{\rm pnr}= |\<{\rm cat}|_{\sqrt k}|\widetilde{\rm cat}\>_{\sqrt k}|^2$ is the fidelity achieved by using the state $|\widetilde{\rm cat}\>_{\sqrt k}$ in Eq. (\ref{eqd11}) with an ideal PNR detector, $F_{\rm mpnr}=\<{\rm cat}|_{\sqrt k}\rho_{\widetilde{\rm cat}}|{\rm cat}\>_{\sqrt k}$ is the fidelity achieved by a MPNR detector. 



Therefore, one can approximately achieve a pure state as that in Eq. (\ref{eqc13}).


\begin{figure*}
\centering\includegraphics[width=0.93\textwidth,trim=2 2 2 2,clip,angle=0]{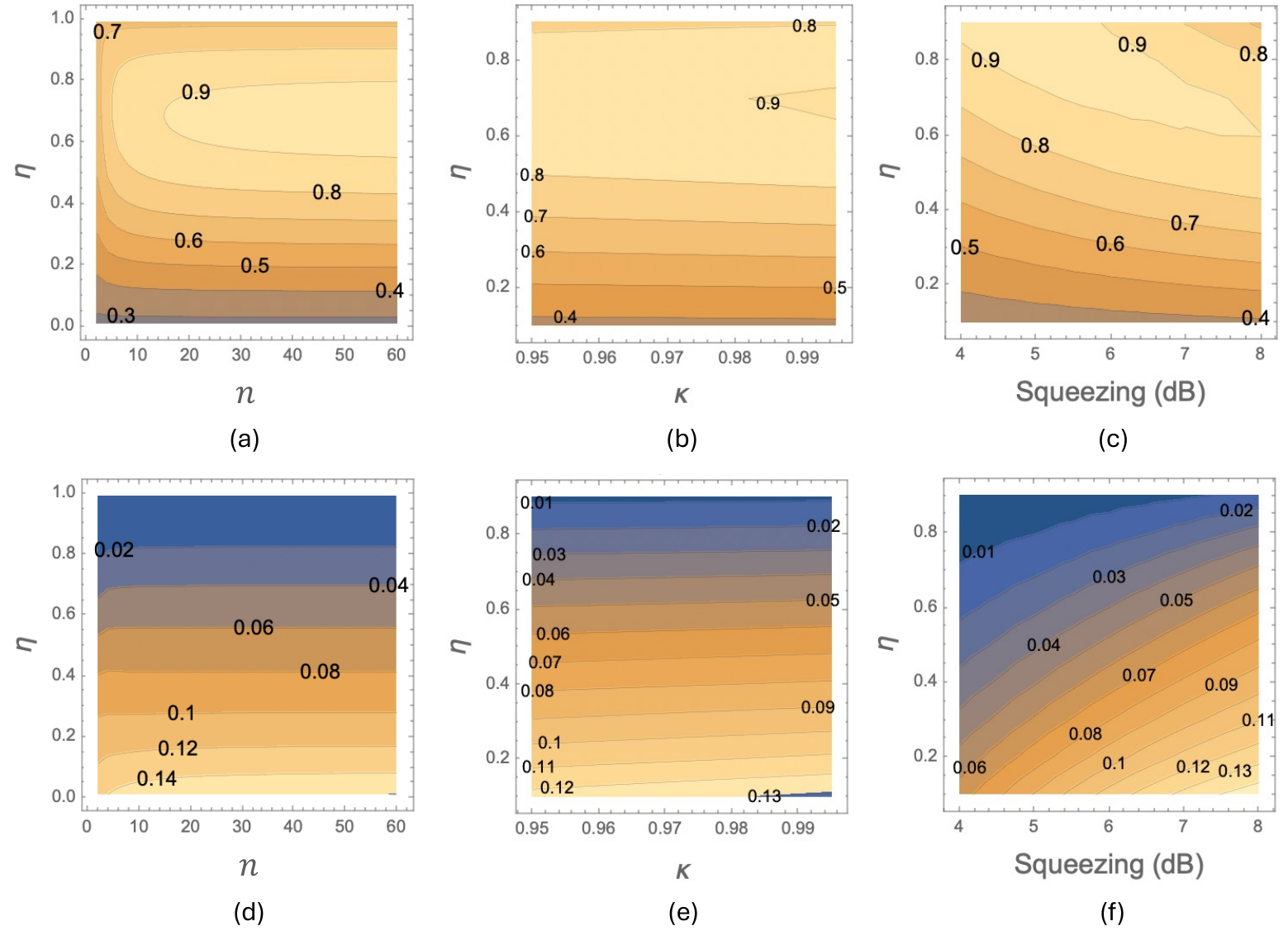}
\caption{Numerical analysis of cat-state generation via breeding protocols. Here we illustrate the two-photon subtraction scenario. 
Fidelity~(a) and success probability~(d) as functions of transmissivity $\eta$ and ON-OFF detector number $n$ (efficiency $\kappa=1$, squeezing 7~dB). Fidelity~(b) and success probability~(e) as functions of transmissivity $\eta$ and efficiency $\kappa$ (squeezing 7~dB, $n=20$ ON-OFF detectors). Fidelity (c) and success probability (f) as functions of transmissivity $\eta$ and squeezing level (efficiency $\kappa=1$, $n=20$ ON-OFF detectors). All data shown in the figures were obtained through numerical simulations using a 32-dimensional truncated space.}
\label{fig:fid_prob_cat_breeding_k=2_2}
\end{figure*}



\subsubsection{Scenario with inefficient ON-OFF detectors}

Without losing the generality, let's assume that all ON-OFF detectors have the same efficiency $\kappa$. In this case, the state in Eq. (\ref{eqc11}) will be changed into the following state: 
\begin{align}
|\Psi'\>= & \int \frac{\d^2 \alpha}{\pi} \frac{1}{\sqrt{\cosh r}} e^{-\tanh r \frac{\alpha^{*2}}{2 }}
e^{-\frac{|\alpha|^2}{2}} \left|\underline{\frac{\sqrt{(1-\eta)\kappa }\,\alpha}{\sqrt n }}\right\>^{\otimes n}\left|\underline{\sqrt {\eta} \, i\alpha}\right\>\left|\underline{\frac{\sqrt{(1-\eta)(1-\kappa) }\,\alpha}{\sqrt n }}\right\>_{\rm env}^{\otimes n},
\end{align}
where $|\cdot\>_{\rm env}$ denote the environment modes. 

After appliying the $n$ ON-OFF detectors, the overall state as follows: 
\begin{align}
\rho_{\widetilde{\rm cat}}\propto &\Tr \left\{ \left\{ \left[\bigoplus_{\sigma_g\in \map S_{n,k}}U_{\sigma_g} \left(|0\>\<0|^{\otimes n-k}\otimes \widetilde I^{\otimes k}\right) U_{\sigma_g}^\dag\right] \otimes I \right\}|\Psi'\>\<\Psi'| \right\}\\
= &\frac{n!}{(n-k)!k!}\cdot \Tr \left[\left( |0\>\<0|^{\otimes n-k}\otimes \widetilde I^{\otimes k} \otimes I  \right)|\Psi'\>\<\Psi'|\right]\\ 
= &\frac{n!}{(n-k)!k!}\cdot \int \frac{\d^2 \alpha\d^2 \beta }{ \pi^2} \frac 1 {\cosh r}e^{-\tanh r \frac{\alpha^{*2}}{2 }}
e^{-\frac{|\alpha|^2}{2}} e^{-\tanh r \frac{\beta^{*2}}{2 }}
e^{-\frac{|\beta|^2}{2}} \nonumber \\
&\times \exp\left[-\frac{(1-\eta)\kappa(|\alpha|^2+|\beta|^2)}{2}\right] \left\{\exp\left[\frac{(1-\eta)\kappa\alpha\beta^*}{n}\right]-1\right\}^k\nonumber \\
&\times \exp\left[-\frac 1 2 (1-\eta)(1-\kappa) \left(|\alpha|^2+|\beta|^2-2\beta^*\alpha\right)\right]\nonumber \\
&\times \left|\underline{\sqrt \eta \, i\alpha}\right\>\left\<\underline{\sqrt \eta \, i\beta}\right|\label{eqc22aa}\\
= &\frac{n!}{(n-k)!k!}\cdot \int \frac{\d^2 \alpha\d^2 \beta }{ \pi^2} \frac 1 {\cosh r}e^{-\tanh r \frac{\alpha^{*2}}{2 }}
e^{-\frac{|\alpha|^2}{2}} e^{-\tanh r \frac{\beta^{*2}}{2 }}
e^{-\frac{|\beta|^2}{2}} \nonumber \\
&\times \exp\left[-\frac{(1-\eta)(|\alpha|^2+|\beta|^2)}{2}\right] \left\{\exp\left[\frac{(1-\eta)\kappa\alpha\beta^*}{n}\right]-1\right\}^k\nonumber \\
&\times \exp\left[ (1-\eta)(1-\kappa) \alpha\beta^*\right]\nonumber \\
&\times \left|\underline{\sqrt \eta \, i\alpha}\right\>\left\<\underline{\sqrt \eta \, i\beta}\right|\label{eqc23ppp}
\end{align}
where Eq. (\ref{eqc22aa}) is derived using the fact that the state $|\Psi'\>$ is invariant under any permutation operation on the first $n$ modes. 

Provided that $1-\kappa$ is a small number, we have 
\begin{align}
\rho_{\widetilde{\rm cat}}\propto 
&\frac{n!}{(n-k)!k!}\cdot \int \frac{\d^2 \alpha\d^2 \beta }{ \pi^2} \frac 1 {\cosh r}e^{-\tanh r \frac{\alpha^{*2}}{2 }}
e^{-\frac{|\alpha|^2}{2}} e^{-\tanh r \frac{\beta^{*2}}{2 }}
e^{-\frac{|\beta|^2}{2}} \nonumber \\
&\times \exp\left[-\frac{(1-\eta)(|\alpha|^2+|\beta|^2)}{2}\right]\left[\frac{(1-\eta)\kappa \alpha\beta^*}{n}\right]^k\nonumber \\
&\times \left[1+\frac{k}{2}\frac{(1-\eta)\kappa\alpha\beta^*}{n}+\frac{(k+3k^2)(1-\eta)^2\kappa^2}{24n^2}(\alpha\beta^*)^2+\mathcal O\left(\frac 1 {n^3}\right)\right]\nonumber \\
&\times \left[1+ (1-\eta)(1-\kappa) \alpha\beta^*+\frac 1 2 (1-\eta)^2(1-\kappa)^2 (\alpha\beta^*)^2+\mathcal O((1-\kappa)^3)\right]\nonumber \\
&\times \left|\underline{\sqrt \eta \, i\alpha}\right\>\left\<\underline{\sqrt \eta \, i\beta}\right|\\
= &\frac{n!}{(n-k)!k!}\cdot \int \frac{\d^2 \alpha\d^2 \beta }{ \pi^2} \frac 1 {\cosh r}e^{-\tanh r \frac{\alpha^{*2}}{2 }}
e^{-\frac{|\alpha|^2}{2}} e^{-\tanh r \frac{\beta^{*2}}{2 }}
e^{-\frac{|\beta|^2}{2}} \nonumber \\
&\times \exp\left[-\frac{(1-\eta)(|\alpha|^2+|\beta|^2)}{2}\right]\left[\frac{(1-\eta)\kappa \alpha\beta^*}{n}\right]^k\nonumber \\
&\times \left\{1+\left[\frac{k}{2}\frac{(1-\eta)\kappa}{n}+(1-\eta)(1-\kappa)\right]\alpha\beta^*\right.\nonumber \\
&+ \left.\left[\frac{(k+3k^2)(1-\eta)^2\kappa^2}{24n^2}+\frac 1 2 (1-\eta)^2(1-\kappa)^2+\frac{k(1-\eta)^2\kappa(1-\kappa)}{2n}\right](\alpha\beta^*)^2\right.\nonumber \\
&+\left[\frac{(k+3k^3)(1-\eta)^3\kappa^2(1-\kappa)}{24n^2}+\frac{k(1-\eta)^3\kappa(1-\kappa)^2}{4n}\right](\alpha\beta^*)^3\nonumber \\
&\left.+\frac{(k+3k^2)(1-\eta)^4\kappa^2(1-\kappa)^2}{48n^2}(\alpha\beta^*)^4+\mathcal O\left(\frac 1 {n^3},(1-\kappa)^3\right)\right\} \left|\underline{\sqrt \eta \, i\alpha}\right\>\left\<\underline{\sqrt \eta \, i\beta}\right|\label{eqc25sss}
\end{align}

Here, we perform a numerical simulation for the cat-state breeding process. In Fig. \ref{fig:fid_prob_cat_breeding_k=2_2}, we also examine the fidelity and success probability as functions of transmisitivity $\eta$, efficiency $\kappa$, and squeezing. The trade-off between fidelity and success probability is illustrated in Fig. \ref{fig:fid_prob_cat_breeding_k=2} of the main text. \QZ{In particular, we perform the evaluation in a truncated Fock basis with dimension $d=32$. 
We represent $\hat a$, $\hat a^\dagger$, and $\hat n$ as $d\times d$ matrices, construct the target cat state $\ket{\mathrm{cat}}_{\sqrt 2}$ and the conditional output vectors $\{\ket{\psi_j}\}$ from Eq.~(\ref{eqc25sss}). 
For each choice of detector number $n\in\{2,4,10,20,\infty\}$ and beamsplitter transmissivity $\tau\in\{1,0.95,0.7\}$ we scan the pixel efficiency $\eta$ on a grid $\eta\in[0.002,0.998]$ with step $0.002$, compute the success probability $P(\eta)=\mathrm{Tr}[\rho_{\rm out}(\eta)]$ and fidelity $F(\eta)=\bra{\mathrm{cat}}\rho_{\rm out}(\eta)\ket{\mathrm{cat}}$ with $\rho_{\rm out}(\eta)$ being the unnormalized form of $\rho_{\widetilde {\rm cat}}$. From the resulting set $\{F(\eta),P(\eta)\}$ we construct the upper envelope $P_{\max}(F)$ by, for each fidelity bin, taking the maximal success probability over all $\eta$ with $F(\eta)\ge F$; these envelopes are the curves plotted in Fig.~\ref{fig:fid_prob_cat_breeding_k=2}. }

\subsection{Cat state breeding with generalized photon subtraction}

It is proved in Ref. \cite{takase2021generation} that one can approximate a generalized Schr{\"o}dinger cat state as: $|{\rm cat}_{\alpha,k}\>=(2(1+(-1)^k\exp(-2|\alpha|^2))^{-1/2}\left[|\alpha\>+(-1)^k|-\alpha\>\right]$ through a generalized photon subtraction process as follows:
\begin{align}
|\widetilde{{\rm cat}_{\sqrt k,k}}\>\propto \left[\<k|\otimes S(-r)\right]U|r\>_{\rm sq}|-r\>_{\rm sq}
\end{align}
where $U$ refers to the beam splitter operation $UaU^\dag = \sqrt {1-\eta} a+\sqrt{\eta}b$, $UbU^\dag = -\sqrt{\eta} a+\sqrt{1-\eta}b$ with $a(b)$ being the annihilation operator of the first (second) mode, $\eta$ being the reflectivity satisfying $0\le \eta \le 1$, 
$|z\>_{\rm sq}=\frac 1 {\sqrt{\cosh r }} \sum_{j=0}^\infty  (-e^{i\phi}\tanh r )^j  \frac{\sqrt{(2j)!}}{2^j j!}|2j\>$ denotes a single-mode squeezed vacuum state.

Using the result in Eq. (\ref{eqc8}), the explicit expression of output state before photon detection can be derived ($k>0$):
\begin{align}
S(r)|\widetilde{{\rm cat}_{\sqrt k,k}}\>\propto & \int \frac{\d^2 \alpha\d^2\beta}{\pi^2} \frac{1}{\cosh r} e^{- \frac{\tanh r(\alpha^{*2}-\beta^{*2})}{2 }}e^{-\frac{|\alpha|^2+|\beta|^2}{2}} \<k\left|\underline{\sqrt{1-\eta }\,\alpha-\sqrt \eta \,\beta}\right\>\left|\underline{\sqrt \eta \, \alpha+\sqrt{1-\eta}\beta}\right\>\\
= & \int \frac{\d^2 \gamma\d^2\delta}{\pi^2} \frac{1}{\cosh r} e^{- \tanh r\frac{(1-2\eta)\gamma^{*2}-(1-2\eta)\delta^{*2}+4\sqrt{\eta(1-\eta)}\gamma^*\delta^*}{2 }}e^{-\frac{|\gamma|^2+|\delta|^2}{2}} \<k\left|\underline{\gamma}\right\>\left|\underline{\delta}\right\>\\ 
= & \sum_{\ell=0}^\infty \int \frac{\d^2 \gamma\d^2\delta}{\pi^2} \frac{1}{\cosh r} e^{- \tanh r\frac{(1-2\eta)\gamma^{*2}-(1-2\eta)\delta^{*2}+4\sqrt{\eta(1-\eta)}\gamma^*\delta^*}{2 }}e^{-|\gamma|^2-|\delta|^2} \frac{\gamma^k \delta^\ell}{\sqrt{k!\ell!}}|\ell\> 
\end{align}

\subsubsection{Example with two-photon subtraction: ideal PNR case}

In the case $k=2$, we have: 
\begin{align}
S(r)|\widetilde{{\rm cat}_{\sqrt 2,2}}\>\propto  & \sum_{\ell=0}^\infty \int \frac{\d^2 \gamma\d^2\delta}{\pi^2} \frac{1}{\cosh r} e^{- \frac{\tanh r(1-2\eta)\gamma^{*2}-\tanh r(1-2\eta)\delta^{*2}+4\tanh r\sqrt{\eta(1-\eta)}\gamma^*\delta^*}{2 }}e^{-|\gamma|^2-|\delta|^2} \frac{\gamma^2 \delta^\ell}{\sqrt{2\ell!}}|\ell\> \\
=&\sum_{\ell=0}^\infty \int \frac{\d^2 \gamma\d^2\delta}{\pi^2} \frac{1}{\cosh r} \left(- \frac{\tanh r(1-2\eta)}{2 }+2(\tanh r)^2\eta(1-\eta)\delta^{*2}\right)
e^{ \frac{\tanh r(1-2\eta)\delta^{*2}}{2 }}e^{-|\gamma|^2-|\delta|^2} \frac{|\gamma|^4 \delta^\ell}{\sqrt{2\ell!}}|\ell\> \\
=&2\sum_{\ell=0}^\infty \int \frac{\d^2\delta}{\pi} \frac{1}{\cosh r} \left(- \frac{\tanh r(1-2\eta)}{2 }+2(\tanh r)^2\eta(1-\eta)\delta^{*2}\right)e^{ \frac{\tanh r(1-2\eta)\delta^{*2}}{2 }}e^{-|\delta|^2} \frac{ \delta^\ell}{\sqrt{2\ell!}}|\ell\> \\
=&2\sum_{\ell=0}^\infty\sum_{j=0}^\infty  \int \frac{\d^2\delta}{\pi} \frac{1}{\cosh r} \left(- \frac{\tanh r(1-2\eta)}{2 }+2(\tanh r)^2\eta(1-\eta)\delta^{*2}\right) \frac{[\tanh r(1-2\eta)]^j\delta^{*2j}}{2^j j! }e^{-|\delta|^2} \frac{ \delta^\ell}{\sqrt{2\ell!}}|\ell\> \\
=&-\frac{\sqrt 2}{2}\sum_{j=0}^\infty  \frac{\tanh r(1-2\eta)}{\cosh r}  \frac{[\tanh r(1-2\eta)]^j}{2^j j! }  \sqrt{(2j)!}|2j\>\\
&+4\sqrt 2\sum_{j=1}^\infty   \frac{\tanh r}{\cosh r} \frac{\eta(1-\eta) j}{(1-2\eta)}\frac{[\tanh r(1-2\eta)]^{j}}{2^j j! }  \sqrt{(2j)!}|2j\>\\
=&\sqrt 2 \sum_{j=0}^\infty   \frac{\tanh r }{\cosh r} \left[\frac{4\eta(1-\eta) j}{1-2\eta}-\frac 1 2 (1-2\eta)\right]\frac{[\tanh r(1-2\eta)]^{j}\sqrt{(2j)!}}{2^j j! }  |2j\>
\end{align}

Here, we conduct a numerical calculation for this breeding process in a 16-dimensional truncated state space. It shows that, with 5 dB input squeezing level, and transmissivity $\eta=0.7815$, one can achieve fidelity $F=0.988$ with a probability of success 5.18$\%$.

\subsubsection{Example with two-photon subtraction: inefficient MPNR case}

Let us look at an MPNR detector associated with a balanced beamsplitter operation. Considering the photon loss at each ON-OFF detector, the output state before post-selection will be: 
\begin{align}
|\Psi'\>= & \int \frac{\d^2 \gamma\d^2\delta}{\pi^2} \frac{1}{\cosh r} e^{- \tanh r\frac{(1-2\eta)\gamma^{*2}-(1-2\eta)\delta^{*2}+4\sqrt{\eta(1-\eta)}\gamma^*\delta^*}{2 }}e^{-\frac{|\gamma|^2+|\delta|^2}{2}} \left|\underline{\sqrt{\frac{\kappa}{ n}}\gamma}\right\>^{\otimes n}\left|\underline{\delta}\right\>\left|\underline{\sqrt{\frac{1-\kappa}{n}}\gamma}\right\>_{\rm env}^{\otimes n},
\end{align}
where $|\psi\>_{\rm env}$ denotes the mode of the environment in pure loss. 

Further, the photon number measurement with ON-OFF detectors gives an outcome $k$. The whole process will project the overall state as follows: 
\begin{align}
\rho_{\widetilde{\rm cat}}\propto &\Tr \left\{ \left\{ \left[\bigoplus_{\sigma_g\in \map S_{n,k}}U_{\sigma_g} \left(|0\>\<0|^{\otimes n-k}\otimes \widetilde I^{\otimes k}\right) U_{\sigma_g}^\dag\right] \otimes I \right\}|\Psi'\>\<\Psi'| \right\}\\
= &\frac{n!}{(n-k)!k!}\cdot \Tr \left[\left( |0\>\<0|^{\otimes n-k}\otimes \widetilde I^{\otimes k} \otimes I  \right)|\Psi'\>\<\Psi'|\right]\\ 
= &\frac{n!}{(n-k)!k!}\cdot \int \frac{\d^2 \gamma\d^2 \delta \d^2 \gamma' \d^2 \delta'}{ \pi^4} \frac 1 {(\cosh r)^2}\nonumber \\
&\times e^{-\tanh r (1-2\eta)\frac{\gamma^{*2}}{2 }}
e^{\tanh r (1-2\eta)\frac{\delta^{*2}}{2 }}e^{-2\sqrt{\eta(1-\eta)}\gamma^*\delta^*}e^{-\frac{|\gamma|^2}{2}} 
e^{-\frac{|\delta|^2}{2}} \nonumber \\
&\times e^{-\tanh r (1-2\eta)\frac{\gamma^{'2}}{2 }}
e^{\tanh r (1-2\eta)\frac{\delta^{'2}}{2 }}e^{-2\sqrt{\eta(1-\eta)}\gamma^{'}\delta^{'}}e^{-\frac{|\gamma'|^2}{2}} 
e^{-\frac{|\delta'|^2}{2}} \nonumber \\
&\times \exp\left[-\frac{|\gamma|^2+|\gamma'|^2}{2}\right] \left\{\exp\left[\frac{\kappa\gamma\gamma^{'*}}{n}\right]-1\right\}^k\exp\left[ (1-\kappa) \gamma^{'*}\gamma\right] \left|\underline{\delta }\right\>\left\<\underline{\, \delta'}\right|\label{eqc40ggg}
\end{align}
where $k=2$ in the two-photon subtraction process. 

Without loss of generality, let's assume the conditions $|\kappa\gamma \gamma^{'*}/n|\ll 1$ and $\left| (1-\kappa) \left(|\gamma|^2+|\gamma'|^2-2\gamma^{'*}\gamma\right)\right|\ll 1$  due to the Gaussian prior probability. Then, we have: 
\begin{align}
\rho_{\widetilde{\rm cat}}\propto 
&\frac{n!}{2(n-2)!}\cdot \int \frac{\d^2 \gamma\d^2 \delta \d^2 \gamma' \d^2 \delta'}{ \pi^4} \frac 1 {(\cosh r)^2}\nonumber \\
&\times e^{-\tanh r (1-2\eta)\frac{\gamma^{*2}}{2 }}
e^{\tanh r (1-2\eta)\frac{\delta^{*2}}{2 }}e^{-2\sqrt{\eta(1-\eta)}\gamma^*\delta^*}e^{-\frac{|\gamma|^2}{2}} 
e^{-\frac{|\delta|^2}{2}} \nonumber \\
&\times e^{-\tanh r (1-2\eta)\frac{\gamma^{'2}}{2 }}
e^{\tanh r (1-2\eta)\frac{\delta^{'2}}{2 }}e^{-2\sqrt{\eta(1-\eta)}\gamma^{'}\delta^{'}}e^{-\frac{|\gamma'|^2}{2}} 
e^{-\frac{|\delta'|^2}{2}} \nonumber \\
&\times \exp\left[-\frac{|\gamma|^2+|\gamma'|^2}{2}\right] \left(\frac{\kappa\gamma\gamma^{'*}}{n}\right)^2\left\{1+\frac{\kappa\gamma\gamma^{'*}}{n}+\frac{7}{12}\left(\frac{\kappa\gamma\gamma^{'*}}{n}\right)^2+\mathcal O\left[ \left(\frac{\kappa\gamma\gamma^{'*}}{n}\right)^3\right]\right\}\nonumber \\
&\times \left[1+ (1-\kappa) \gamma^{'*}\gamma+\frac 1 2 (1-\kappa)^2 (\gamma^{'*}\gamma)^2+ \mathcal O((1-\kappa)^3)\right] \left|\underline{\delta }\right\>\left\<\underline{\, \delta'}\right|\\
= &\frac{n!\kappa^2}{n^2(n-2)!}\left\{|\Psi_0\>\<\Psi_0|+\left(\frac \kappa n+1-\kappa\right) |\Psi_1\>\<\Psi_1|+\left[\frac{\kappa(1-\kappa)}{n}+\frac{7\kappa^2}{12n^2}+\frac {(1-\kappa)^2}{2}\right]|\Psi_2\>\<\Psi_2|\right.\nonumber \\
&\left.+\left[\frac{7\kappa^2(1-\kappa)}{12n^2}+\frac{\kappa(1-\kappa)^2}{2n}\right]|\Psi_3\>\<\Psi_3| +\frac{7\kappa^2(1-\kappa)^2}{24n^2}|\Psi_4\>\<\Psi_4|+\mathcal O\left(\frac{1}{n^3},(1-\kappa)^3\right)\right\}
\end{align}
where the explicit expressions for the vectors $\{|\Psi_0\>,\cdots,|\Psi_4\>\}$ are as follows: 
\begin{align}
|\Psi_0\>= & \int \frac{\d^2 \gamma\d^2\delta}{\pi^2} \frac{1}{\cosh r} e^{- \tanh r\frac{(1-2\eta)\gamma^{*2}-(1-2\eta)\delta^{*2}+4\sqrt{\eta(1-\eta)}\gamma^*\delta^*}{2 }}e^{-\frac{|\gamma|^2+|\delta|^2}{2}} \<2\left|\underline{\gamma}\right\>\left|\underline{\delta}\right\>\\ 
=&\sqrt 2 \sum_{j=0}^\infty   \frac{\tanh r }{\cosh r} \left[\frac{4\eta(1-\eta) j}{1-2\eta}-\frac 1 2 (1-2\eta)\right]\frac{[\tanh r(1-2\eta)]^{j}\sqrt{(2j)!}}{2^j j! }  |2j\>
\end{align}

\begin{align}
|\Psi_1\>= & \int \frac{\d^2 \gamma\d^2\delta}{\pi^2} \frac{1}{\cosh r} e^{- \tanh r\frac{(1-2\eta)\gamma^{*2}-(1-2\eta)\delta^{*2}+4\sqrt{\eta(1-\eta)}\gamma^*\delta^*}{2 }}e^{-\frac{|\gamma|^2+|\delta|^2}{2}}\gamma  \<2\left|\underline{\gamma}\right\>\left|\underline{\delta}\right\>\\
=&   \sum_{\ell=0}^\infty \int \frac{\d^2 \gamma\d^2\delta}{\pi^2} \frac{1}{\cosh r} e^{- \frac{\tanh r(1-2\eta)\gamma^{*2}-\tanh r(1-2\eta)\delta^{*2}+4\tanh r\sqrt{\eta(1-\eta)}\gamma^*\delta^*}{2 }}e^{-|\gamma|^2-|\delta|^2} \frac{\gamma^3 \delta^\ell}{\sqrt{2\ell!}}|\ell\> \\
=&   \sum_{\ell=0}^\infty \int \frac{\d^2 \gamma\d^2\delta}{\pi^2} \frac{1}{\cosh r} \left(\frac{-4(\tanh r \sqrt{\eta(1-\eta)})^3\delta^{*3}}{3}+(1-2\eta)\sqrt{\eta(1-\eta)}(\tanh r)^2\delta^* \right)\nonumber \\
&\times e^{\frac{\tanh r(1-2\eta)\delta^{*2}}{2 }}e^{-|\gamma|^2-|\delta|^2} \frac{|\gamma|^6 \delta^\ell}{\sqrt{2\ell!}}|\ell\> \\
=& 6  \sum_{\ell=0}^\infty \int \frac{\d^2\delta}{\pi} \frac{1}{\cosh r} \left(\frac{-4(\tanh r \sqrt{\eta(1-\eta)})^3\delta^{*3}}{3}+(1-2\eta)\sqrt{\eta(1-\eta)}(\tanh r)^2\delta^* \right)\nonumber \\
&\times e^{\frac{\tanh r(1-2\eta)\delta^{*2}}{2 }}e^{-|\delta|^2} \frac{ \delta^\ell}{\sqrt{2\ell!}}|\ell\> \\
=& 6  \sum_{\ell=0}^\infty \int \frac{\d^2\delta}{\pi} \frac{1}{\cosh r} \left(\frac{-4(\tanh r \sqrt{\eta(1-\eta)})^3\delta^{*3}}{3}+(1-2\eta)\sqrt{\eta(1-\eta)}(\tanh r)^2\delta^* \right)\nonumber \\
&\times \sum_{j=0}\frac{(\tanh r(1-2\eta))^j\delta^{*2j}}{2^j j!} e^{-|\delta|^2} \frac{ \delta^\ell}{\sqrt{2\ell!}}|\ell\> \\
=& 6 \frac{1}{\cosh r} \frac{-4(\tanh r \sqrt{\eta(1-\eta)})^3}{3}\sum_{j=0}^\infty \frac{(\tanh r(1-2\eta))^j}{2^j j!} \frac{ \sqrt{(2j+3)!}}{\sqrt{2}}|2j+3\> \\
&+ 6    \frac{1}{\cosh r} (1-2\eta)\sqrt{\eta(1-\eta)}(\tanh r)^2 \sum_{j=0}^\infty \frac{(\tanh r(1-2\eta))^j}{2^j j!}  \frac{ \sqrt{(2j+1)!}}{\sqrt{2}}|2j+1\> \\
=& 6 \sum_{j=1}^\infty \frac{1}{\cosh r} \frac{-8(\tanh r \sqrt{\eta(1-\eta)})^3j}{3(\tanh r (1-2\eta))}\frac{(\tanh r(1-2\eta))^{j}}{2^{j} j!} \frac{ \sqrt{(2j+1)!}}{\sqrt{2}}|2j+1\> \\
&+ 6   \sum_{j=0}^\infty  \frac{1}{\cosh r} (1-2\eta)\sqrt{\eta(1-\eta)}(\tanh r)^2 \frac{(\tanh r(1-2\eta))^j}{2^j j!}  \frac{ \sqrt{(2j+1)!}}{\sqrt{2}}|2j+1\> \\
=&  3\sqrt{2}   \sum_{j=0}^\infty  \frac{1}{\cosh r} \left[(1-2\eta)\sqrt{\eta(1-\eta)}(\tanh r)^2 - \frac{8(\tanh r \sqrt{\eta(1-\eta)})^3j}{3(\tanh r (1-2\eta))}\right]\frac{(\tanh r(1-2\eta))^j}{2^j j!}   \sqrt{(2j+1)!}|2j+1\> 
\end{align}

\begin{align}
|\Psi_2\>= & \int \frac{\d^2 \gamma\d^2\delta}{\pi^2} \frac{1}{\cosh r} e^{- \tanh r\frac{(1-2\eta)\gamma^{*2}-(1-2\eta)\delta^{*2}+4\sqrt{\eta(1-\eta)}\gamma^*\delta^*}{2 }}e^{-\frac{|\gamma|^2+|\delta|^2}{2}}\gamma^2  \<2\left|\underline{\gamma}\right\>\left|\underline{\delta}\right\>\\
=&   \sum_{\ell=0}^\infty \int \frac{\d^2 \gamma\d^2\delta}{\pi^2} \frac{1}{\cosh r} e^{- \frac{\tanh r(1-2\eta)\gamma^{*2}-\tanh r(1-2\eta)\delta^{*2}+4\tanh r\sqrt{\eta(1-\eta)}\gamma^*\delta^*}{2 }}e^{-|\gamma|^2-|\delta|^2} \frac{\gamma^4 \delta^\ell}{\sqrt{2\ell!}}|\ell\> \\
=&   \sum_{\ell=0}^\infty \int \frac{\d^2 \gamma\d^2\delta}{\pi^2} \frac{1}{\cosh r}  \left[\frac{(-2\tanh r\sqrt{\eta(1-\eta)})^4\delta^{*4}}{24}-(\tanh r)^3 (1-2\eta)\eta(1-\eta)\delta^{*2}+\frac{(\tanh r )^2(1-2\eta)^2}{8}\right]\nonumber \\
&\times e^{ \frac{\tanh r(1-2\eta)\delta^{*2}}{2 }}e^{-|\gamma|^2-|\delta|^2} \frac{|\gamma|^8 \delta^\ell}{\sqrt{2\ell!}}|\ell\> 
\end{align}

\begin{align}
=&  24 \sum_{\ell=0}^\infty \int \frac{\d^2\delta}{\pi} \frac{1}{\cosh r} \nonumber \\
&\times \left[\frac{(-2\tanh r\sqrt{\eta(1-\eta)})^4\delta^{*4}}{24}-(\tanh r)^3 (1-2\eta)\eta(1-\eta)\delta^{*2}+\frac{(\tanh r )^2(1-2\eta)^2}{8}\right]\nonumber \\
&\times e^{ \frac{\tanh r(1-2\eta)\delta^{*2}}{2 }}e^{-|\delta|^2} \frac{\delta^\ell}{\sqrt{2\ell!}}|\ell\> \\
=&  24  \sum_{j=2}^\infty \int \frac{\d^2\delta}{\pi} \frac{1}{\cosh r}  \left[\frac{(-2\tanh r\sqrt{\eta(1-\eta)})^4}{24(\tanh r (1-2\eta))^2}\right]4(j-1)j \frac{(\tanh r(1-2\eta))^{j}}{2^{j} j! }e^{-|\delta|^2} \frac{|\delta|^{4j}}{\sqrt{2(2j)!}}|2j\> \\
&+  24 \sum_{j=1}^\infty \int \frac{\d^2\delta}{\pi} \frac{1}{\cosh r}  \frac{-(\tanh r)^3 (1-2\eta)\eta(1-\eta)}{\tanh r(1-2\eta)}2j  \frac{(\tanh r(1-2\eta))^{j}}{2^{j} j! }e^{-|\delta|^2} \frac{|\delta|^{4j}}{\sqrt{2(2j)!}}|2j\> \\
&+  24 \sum_{j=0}^\infty\int \frac{\d^2\delta}{\pi} \frac{1}{\cosh r} \left[\frac{(\tanh r )^2(1-2\eta)^2}{8}\right]  \frac{(\tanh r(1-2\eta))^j}{2^j j! }e^{-|\delta|^2} \frac{|\delta|^{4j}}{\sqrt{2(2j)!}}|2j\> \\
=
&  12\sqrt 2 \sum_{j=0}^\infty\frac{(\tanh r)^2}{\cosh r} \left[\frac{(1-2\eta)^2}{8}-\eta(1-\eta)2j +\frac{8\eta^2(1-\eta)^2}{3(1-2\eta)^2}(j-1)j\right]  \frac{(\tanh r(1-2\eta))^j}{2^j j! }\sqrt{(2j)!}|2j\> 
\end{align}

\begin{align}
|\Psi_3\>= & \int \frac{\d^2 \gamma\d^2\delta}{\pi^2} \frac{1}{\cosh r} e^{- \tanh r\frac{(1-2\eta)\gamma^{*2}-(1-2\eta)\delta^{*2}+4\sqrt{\eta(1-\eta)}\gamma^*\delta^*}{2 }}e^{-\frac{|\gamma|^2+|\delta|^2}{2}}\gamma^3  \<2\left|\underline{\gamma}\right\>\left|\underline{\delta}\right\>\\
=&   \sum_{\ell=0}^\infty \int \frac{\d^2 \gamma\d^2\delta}{\pi^2} \frac{1}{\cosh r} e^{- \frac{\tanh r(1-2\eta)\gamma^{*2}-\tanh r(1-2\eta)\delta^{*2}+4\tanh r\sqrt{\eta(1-\eta)}\gamma^*\delta^*}{2 }}e^{-|\gamma|^2-|\delta|^2} \frac{\gamma^5 \delta^\ell}{\sqrt{2\ell!}}|\ell\> \\
=&  120 \sum_{\ell=0}^\infty \int \frac{\d^2\delta}{\pi} \frac{1}{\cosh r} \nonumber \\
&\times \left[\frac{(-2\tanh r \sqrt{\eta(1-\eta)})^5 \delta^{*5}}{120}+ \frac{2(\tanh r)^4(1-2\eta)( \sqrt{\eta(1-\eta)})^3\delta^{*3}}{3}-\frac{(\tanh r)^3(1-2\eta)^2\sqrt{\eta(1-\eta)}\delta^*}{4}\right]\nonumber \\
&\times e^{- \frac{-\tanh r(1-2\eta)\delta^{*2}}{2 }}e^{-|\delta|^2} \frac{ \delta^\ell}{\sqrt{2\ell!}}|\ell\> \\
=&  120  \sum_{j=0}^\infty  \frac{(\tanh r)^3}{\cosh r}  \left[\frac{-16( \sqrt{\eta(1-\eta)})^5 }{15(1-2\eta)^2}(j-1)j+ \frac{4( \sqrt{\eta(1-\eta)})^3}{3}j-\frac{(1-2\eta)^2\sqrt{\eta(1-\eta)}}{4}\right]\nonumber \\
&\times  \frac{(\tanh r(1-2\eta))^j}{2^j j! } \frac{ \sqrt{(2j+1)!}}{\sqrt{2}}|2j+1\> 
\end{align}

\begin{align}
|\Psi_4\>= & \int \frac{\d^2 \gamma\d^2\delta}{\pi^2} \frac{1}{\cosh r} e^{- \tanh r\frac{(1-2\eta)\gamma^{*2}-(1-2\eta)\delta^{*2}+4\sqrt{\eta(1-\eta)}\gamma^*\delta^*}{2 }}e^{-\frac{|\gamma|^2+|\delta|^2}{2}}\gamma^4  \<2\left|\underline{\gamma}\right\>\left|\underline{\delta}\right\>\\
=&   \sum_{\ell=0}^\infty \int \frac{\d^2 \gamma\d^2\delta}{\pi^2} \frac{1}{\cosh r} e^{- \frac{\tanh r(1-2\eta)\gamma^{*2}-\tanh r(1-2\eta)\delta^{*2}+4\tanh r\sqrt{\eta(1-\eta)}\gamma^*\delta^*}{2 }}e^{-|\gamma|^2-|\delta|^2} \frac{\gamma^6 \delta^\ell}{\sqrt{2\ell!}}|\ell\> \\
=&   720\sum_{\ell=0}^\infty \int \frac{\d^2\delta}{\pi} \frac{1}{\cosh r} \left[\frac{4 (\tanh r )^6\eta^3(1-\eta)^3\delta^{*6}}{45}-\frac{ (\tanh r )^5 (1-2\eta)\eta^2(1-\eta)^2\delta^{*4}}{3}\right.\nonumber \\
&+\left.\frac{(\tanh r)^4(1-2\eta)^2\eta (1-\eta)\delta^{*2}}{4}-\frac{(\tanh r )^3 (1-2\eta)^3}{48}\right] e^{\frac{\tanh r(1-2\eta)\delta^{*2}}{2 }}e^{-|\delta|^2} \frac{ \delta^\ell}{\sqrt{2\ell!}}|\ell\> 
\end{align}

\begin{align}
=&   720\sum_{j=0}^\infty   \frac{1}{\cosh r} \left[\frac{4 (\tanh r )^6\eta^3(1-\eta)^3\delta^{*6}}{45(\tanh r)^3(1-2\eta)^3}8(j-2)(j-1)j-\frac{ (\tanh r )^5 (1-2\eta)\eta^2(1-\eta)^2\delta^{*4}}{3(\tanh r )^2(1-2\eta)^2}4(j-1)j\right.\nonumber \\
&+\left.\frac{(\tanh r)^4(1-2\eta)^2\eta (1-\eta)\delta^{*2}}{4(\tanh r )(1-2\eta)}2j-\frac{(\tanh r )^3 (1-2\eta)^3}{48}\right] \frac{(\tanh r)^j(1-2\eta)^j}{2^j j! } \frac{ \sqrt{(2j)!}}{\sqrt{2}}|2j\> \\
=&   720\sum_{j=0}^\infty   \frac{(\tanh r)^3}{\cosh r} \left[\frac{32 \eta^3(1-\eta)^3\delta^{*6}}{45(1-2\eta)^3}(j-2)(j-1)j-\frac{4  \eta^2(1-\eta)^2\delta^{*4}}{3(1-2\eta)}(j-1)j\right.\nonumber \\
&+\left.\frac{(1-2\eta)\eta (1-\eta)\delta^{*2}}{2}j-\frac{(1-2\eta)^3}{48}\right] \frac{(\tanh r)^j(1-2\eta)^j}{2^j j! } \frac{ \sqrt{(2j)!}}{\sqrt{2}}|2j\> 
\end{align}

\section{Modeling of correlated detectors}\label{app:correlated detectors}

\subsection{Definition: two-mode error}

Consider a scenario in which a two-mode state  simultaneously click two adjacent ON-OFF detectors. In the ideal scenario, the input-output relation is supposed to be: 
\begin{align}\label{eqe1w}
\begin{cases}
|0\>\otimes |0\>&\to 0\\
|0\>\otimes |k\> &\to 1\\
|k\>\otimes |0\> &\to 1\\
|k\>\otimes |k\> &\to 2
\end{cases},\ \ \ \ \ \forall k \ge 1.
\end{align} 
If the two detectors are correlated, the input-output relation will be: 
\begin{align}\label{eqe2w}
\begin{cases}
|0\>\otimes |0\>&\to 0\\
|0\>\otimes |k\> &\to 2\\
|k\>\otimes |0\> &\to 2\\
|k\>\otimes |k\> &\to 2
\end{cases},\ \ \ \ \ \forall k \ge 1.
\end{align}
In this case, the corresponding POVM for estimating total photon number of two modes is: 
\begin{align}
 M_i'&= |0\>\<0|\otimes |0\>\<0|\\
 M_{ii}'&= I \otimes I - |0\>\<0|\otimes |0\>\<0|. 
\end{align}
The observable of the total photon number is: 
\begin{align}
\widehat N^{'}_2&=2  \left(I \otimes I - |0\>\<0|\otimes |0\>\<0|\right). 
\end{align}

Considering the scenario where the correlation error occurs with a small probability $p$, the observed photon number is: 
\begin{align}
\widetilde N_{\rm ON-OFF, cor,2}^h&=(1-p) \<\widehat N_{\rm ON-OFF,2}^h\>+p \<\widehat N_2^{'h}\>
\end{align}
where $\widehat N_{\rm ON-OFF,2}=2 \widetilde I \otimes \widetilde I+ \widetilde I \otimes |0\>\<0|+ |0\>\<0|\otimes \widetilde I$ denotes the observable of ON-OFF detectors with $\widetilde I=I-|0\>\<0|$ being a projection onto the non-vacuum state space. Accordingly, one can defined the observable $\widehat N^h_{\rm ON-OFF, cor,2}=(1-p) \widehat N^h_{\rm ON-OFF,2}+p \widehat N^{'h}_2$ for this error.

\subsection{Multimode case}

Consider a scenario with $n$ ON-OFF detectors. Without losing the generality, let's assume that the number of $n$ is even. Furthermore, let's assume that the $n$ ON-OFF detectors can be partitioned into $n/2$ adjacent pairs. Then, the observed  photon number moment will be: 
\begin{align}
\widetilde  N_{\rm ON-OFF, cor,n}=&\left\<\sum_{j=1}^{n/2}\widehat N_{\rm ON-OFF, cor,2,(j)}^h\bigotimes_{k=1\atop k\neq j}^{n/2} I_{2k-1}\otimes I_{2k} \right\>\\
=&\left\<\sum_{j=1}^{n/2}\widehat N_{\rm ON-OFF,2,(j)}\bigotimes_{k=1\atop k\neq j}^{n/2} I_{2k-1}\otimes I_{2k} \right\>\\
&-p\left\<\sum_{j=1}^{n/2}\left[\widehat N_{\rm ON-OFF,2,(j)}-\widehat N_{2,(j)}^{'}\right]\bigotimes_{k=1\atop k\neq j}^{n/2} I_{2k-1}\otimes I_{2k} \right\>\\
=&\left\<\sum_{j=1}^{n/2}\widehat N_{\rm ON-OFF,2,(j)}\bigotimes_{k=1\atop k\neq j}^{n/2} I_{2k-1}\otimes I_{2k} \right\>\\
&+p\left\<\sum_{j=1}^{n/2}\left(\widetilde I_{2j-1}\otimes |0\>\<0|_{2j}+|0\>\<0|_{2j-1}\otimes \widetilde I_{2j}\right)\bigotimes_{k=1\atop k\neq j}^{n/2} I_{2k-1}\otimes I_{2k} \right\>
\end{align}

Given a state with P-function $P(\alpha)$, and applying a balanced beamsplitter, one obtain the following observed photon with a bias vanishing in $\mathcal O(p)$: 
\begin{align}
\widetilde N_{\rm ON-OFF,cor,n}
&=\widetilde N_{\rm ON-OFF,n}-np \int \frac{\d^2 \alpha}{\pi} P(\alpha)  \left[1-\exp\left(-\frac{|\alpha|^2}{n}\right)\right]\exp\left(-\frac{|\alpha|^2}{n}\right) \\
&=\widetilde N_{\rm ON-OFF,n}-np\int \frac{\d^2 \alpha}{\pi} P(\alpha)  \left[1-\exp\left(-\frac{|\alpha|^2}{n}\right)\right]\left[1-\frac{|\alpha|^2}{n}+\mathcal O\left(\frac{|\alpha|^2}{n^2}\right)\right] \\
&=  \widetilde N_{\rm ON-OFF,n}+\mathcal O\left(p\right)+\mathcal O\left(\frac {p} n\right).
\end{align}

\end{widetext}

\end{document}